\documentclass{article}

\usepackage[frenchb,english]{babel}

\usepackage{amsmath}
\usepackage{amssymb}
\usepackage{amsthm}
\usepackage{amsfonts} 
\usepackage{stmaryrd}
\usepackage{graphicx}
\usepackage{comment}
\usepackage{dsfont}
\usepackage{xcolor}
\usepackage{mathtools}
\usepackage{bigints}
\usepackage{yhmath}
\usepackage{enumitem} 
\usepackage{subcaption} 
\usepackage{appendix} 
\usepackage{geometry} 
\usepackage{authblk} 
\usepackage{empheq} 
\usepackage{tikz} 

\newtheorem{lemma}{Lemma}
\newtheorem{propo}{Proposition}
\newtheorem{theor}{Theorem}

\newtheorem{defin}{Definition}
\newtheorem{remar}{Remark}

\newtheorem{notat}{Notation}

\newcommand{\vertii}[1]{{\left\vert\kern-0.3ex\left\vert #1 
    \right\vert\kern-0.3ex\right\vert}}

\newcommand{\vertiii}[1]{{\left\vert\kern-0.3ex\left\vert\kern-0.3ex\left\vert #1 
    \right\vert\kern-0.3ex\right\vert\kern-0.3ex\right\vert}}
    
\newcommand*\dd{\mathop{}\!\mathrm{d}}

\title{An Alexander-like theorem for a particle model with inelastic collisions}
\author[1]{Th\'eophile Dolmaire}
\author[2]{Juan J. L. Vel\'azquez}
\affil[1]{Dipartimento di Ingegneria e Scienze dell’Informazione e Matematica (DISIM), Università degli Studi dell’Aquila, Edificio Renato Ricamo, via Vetoio, Coppito, 67100 L’Aquila, Italy}
\affil[2]{Institute for Applied Mathematics, University of Bonn, Endenicher Allee 60, D-53115 Bonn, Germany}

\begin{document}

\maketitle

\begin{abstract}
\noindent
We consider a finite system of hard spheres that collide inelastically according to a particular model, losing a fixed amount of kinetic energy at each collision. We develop the theory of the Transport-Collision-Transport (TCT) dynamics, which allows to study precisely the evolution of the Lebesgue measure in the phase space under the action of the flows of particle systems that can interact via instantaneous binary collision. We show that the scattering mapping associated to the inelastic hard sphere system we introduce preserves locally the Lebesgue measure in the velocity space, in spite of the fact that a positive amount of kinetic energy is lost at each inelastic collision. We prove the analog of Alexander's theorem for our model, which allows us to deduce the global well-posedness of the trajectories, for almost every initial datum.
\end{abstract}


\section{Introduction}
\numberwithin{equation}{section}
In this article we consider a model of $N$ inelastic hard spheres, defined such that a fixed amount $\varepsilon_0 > 0$ of kinetic energy is dissipated during each collision that is energetic enough. Our main objective will be to establish the well-posedness of the dynamics of such a particle system, using in particular the properties of the flow on the phase space of such a particle system, and how such a flow acts on the Lebesgue measure in the phase space.\\
Considering such a model can be motivated by the study of atoms that interact with photons, see for instance \cite{Oxen986} for a general introduction to these particle systems interacting with radiation. More specifically, when a collision between two excited atoms takes place, a deexcitation can occur, giving rise to the emission of one or several photons. If we consider also that the energy transported by such photons is quantized, as prescribed by quantum mechanics, the kinetic energy that is lost during the collision of the two atoms belongs to a countable set of positive numbers.\\
Such an emission process is in general stochastic, in the sense that it takes place according to a certain probability. Moreover, if the photons remain in the system, the inverse process, i.e. an absorption of a photon during a collision, can also take place. Therefore, in general, a collision can change the level of excitation of one of the colliding atoms. Let us also observe that, for instance in \cite{Oxen986}, \cite{RoSM997}, \cite{MoPR998}, or in the more recent works \cite{JaVe022} and \cite{Dema023}, the authors study systems of atoms that interact between themselves, and interact with photons via emission and absorption processes. In the models considered in \cite{RoSM997}, \cite{MoPR998}, \cite{JaVe022}, \cite{Dema023}, the particles can have only two level of internal energy (ground state, or excited). Therefore these models with two interacting species, when keeping track of the evolution of the photons, can be described as a single system, for which the kinetic energy is globally preserved. However, if the photons are not kept in the description (for instance due to the fact that they escape to infinity after their emission), the system of atoms loses kinetic energy along its evolution.\\
In the present article, we will consider the latter framework. On the one hand, we will not describe the photons, and focus only on the atoms, that will be represented by identical hard spheres. On the other hand, we will consider a simple version of the emission process during collisions, namely, if the relative velocity $v-v_*$ is large enough, then an inelastic collision takes place. Otherwise, if the collision is not energetic enough, we will assume that an elastic collision takes place, and in this case we will describe the evolution of the particles according to the rules of the usual elastic hard sphere system. Notice that the process defined in this way is deterministic, but it is also irreversible, because the kinetic energy that is lost during an inelastic collision will never be recovered later by the system of hard spheres. Finally, let us observe that each particle will be described only by its position and its velocity. In particular, we will not attach to the particles their respective excitation level: the emission of a photon, or not, depend only on the relative velocity of two colliding particles.\\
We will refer to this particle system as the \emph{inelastic hard sphere system with emission} (or IHSE, in short). Such a particle system exhibits a quite surprising property in dimension $d=2$: on the one hand, by definition, during each collision which is sufficiently energetic the system will lose a positive amount of kinetic energy, on the other hand we will show that the scattering mapping will preserve the measure in the velocity space. Let us recall that the scattering mapping is the application that to a pair of pre-collisional velocities $(v,v_*)$ of colliding particles associates the post-collisional velocities $(v',v'_*)$. Intuitively, these two properties look contradictory, we will actually show that both of them can hold true simultaneously, at least in dimension $d=2$.\\
\newline
The main issue that we will consider in the present article concerns the well-posedness of particle systems, for almost every initial configuration, of hard sphere systems with inelastic interaction. To the best of our knowledge, there is no result of this type in the literature. As a matter of fact, this type of well-posedness result is still not obtained for the usual inelastic hard sphere system, with a fixed restitution coefficient. In that system, it is assumed that, during a collision, a fixed fraction of the normal component 
of the relative velocity is lost, while the tangential relative velocity is preserved. In other words, the post-collisional relative velocity $v'-v'_*$ satisfies:
\begin{align}
\label{EQUAT_Loi_ColliCoeffRestiRFixe}
(v'-v'_*)\cdot\omega = - r(v-v_*)\cdot\omega,
\end{align}
where $r$, the restitution coefficient, is a \emph{fixed} positive number ($0 < r < 1$). One major difficulty with this model is the onset of the phenomenon of inelastic collapse, characterized by an infinite number of collisions in finite time. See for instance \cite{McYo993}, \cite{ZhKa996}, or the recent works \cite{DoVeAr1} and \cite{DoVeAr2} concerning the inelastic collapse. Let us mention that this model is extensively used to describe granular gases, for instance through an inelastic version of the Boltzmann equation (\cite{BrPo004}, \cite{CHMR021}).\\
However, in the case of elastic hard spheres, the well-posedness of such a particle system is established: this is Alexander's theorem (\cite{Alex975}, \cite{Alex976}, see also \cite{GSRT013} for a modern presentation). The main difficulty to prove such a result lies in the fact that collisions involving three particles or more can take place, preventing to define the dynamics further. 
Nevertheless, Alexander obtained the well-posedness for \emph{almost every} initial configuration, with respect to the Lebesgue measure. This result is then the starting point to study the derivation of the Boltzmann equation from a system of $N$ hard spheres, completed by the celebrated Lanford's theorem \cite{Lanf975}. The reader may also refer to the more recent \cite{GSRT013} for a modern proof of Alexander's theorem, and to \cite{PuSS014} for further developments on Lanford's theorem, including more general interaction potentials for the particle system.\\
One of the key arguments in the proof of Alexander's theorem is the fact that the hard sphere transport preserves the Lebesgue measure. Since, in dimension $d=2$, the scattering of the inelastic hard sphere with emission system shares also this property, it is a natural question to study if the arguments of Alexander can be adapted to our model. As a matter of fact, we can obtain a similar result, namely, Theorem \ref{THEORSect3AlexanderIHSWEM}, which is the main result of this article. To the best of our knowledge, this is the first result in the literature of global well-posedness concerning inelastic particle systems. We emphasize that the result holds only in dimension $d=2$.\\
The plan of the article is the following. In Section \ref{SECTI__2__}, we introduce in detail the inelastic hard sphere with emission system. An important step to study the evolution of the Lebesgue measure under the action of the flow of particle systems is to restrict the dynamics on time intervals such that at most one collision can take place, ensuring that such a dynamics is indeed well-defined. In Section \ref{SSECT_2.2_} we introduce the associated notion of TCT dynamics, describing such one-collision dynamics. In Section \ref{SECTI__3__} we obtain a general formula concerning the Jacobian determinant of the TCT dynamics, which allows to estimate how the Lebesgue measure evolves locally under the action of the flows of particle systems. This result is presented in Theorem \ref{THEOREvoluMesur_TCT_Proce}, page \pageref{THEOREvoluMesur_TCT_Proce}. In Section \ref{SECTI__4__}, we apply Theorem \ref{THEOREvoluMesur_TCT_Proce}. On the one hand, we apply the result to the elastic hard sphere system, recovering the well-known result that the flow of such a system preserves indeed the Lebesgue measure in the phase space. On the other hand, we apply also the formula to the inelastic hard sphere with emission system we introduced. In particular, we prove that, although the scattering mapping of such a system preserves the Lebesgue measure, the flow does not. In Section \ref{SECTI__5__}, we use the previous result to prove an Alexander-type theorem for the inelastic hard sphere with emission system, in dimension $d=2$. In this final section, we provide also a geometrical interpretation of the fact that the scattering of the inelastic hard sphere with emission preserves the measure in dimension $d=2$. Besides, we present examples, in dimension $d \geq 2$, with $d$ arbitrary, of scattering mappings describing inelastic collisions, dissipating a fixed amount of kinetic energy, but preserving the measure in the phase space.

\section{The inelastic hard spheres with emission}
\label{SECTI__2__}

\subsection{The model}
\label{SSec1PresentatiModèl}

We consider a system of $N$ spherical particles, of diameter $1$, evolving in $\mathbb{R}^d$, where $d \geq 2$ is a positive integer. The position and the velocity, of each of the particles will be denoted respectively by $x_i \in \mathbb{R}^d$ and $v_i \in \mathbb{R}^d$. The \emph{configuration} of the system is the vector
\begin{align}
Z_N = \left(x_1,v_1,\dots,x_N,v_N\right) \in \mathbb{R}^{2dN},
\end{align}
collecting the positions and velocities of all of the $N$ particles of the system.\\
In order to describe the collisions, we introduce the set $\mathcal{P}_N$ of pairs of ordered indices:
\begin{align}
\mathcal{P}_N = \left\{ (i,j) \in \{1,\dots,N\}^2\ /\ 1 \leq i < j \leq N \right\}.
\end{align}
We introduce also the two following spaces of configurations.

\begin{defin}[Phase space]
Let $d \geq 2$ and $N$ be two positive integers. We define the \emph{phase space}, denoted by $\overline{\mathcal{D}}_N$, as the following subset of $\mathbb{R}^{2dN}$:
\begin{align}
\overline{\mathcal{D}}_N = \big\{ Z_N \in \left(x_1,v_1,\dots,x_N,v_N\right) \in \mathbb{R}^{2dN}\ /\ \forall\, (i,j) \in \mathcal{P}_N, \ \vert x_i-x_j \vert \geq 1 \big\}.
\end{align}
In the same way, we define the set $\mathcal{D}_N$, which is the interior of $\overline{\mathcal{D}}_N$, as the following subset of $\mathbb{R}^{2dN}$:
\begin{align}
\mathcal{D}_N = \big\{ Z_N \in \left(x_1,v_1,\dots,x_N,v_N\right) \in \mathbb{R}^{2dN}\ /\ \forall\, (i,j) \in \mathcal{P}_N, \ \vert x_i-x_j \vert > 1 \big\}.
\end{align}
\end{defin}

\begin{remar}
The set $\mathcal{D}_N$ describes the configurations of $N$ spherical particles of radius $1$ evolving in the Euclidean space $\mathbb{R}^d$, such that there is no contact between the particles. The configuration $Z_N$ is describing a collision between at least two particles when $Z_N \in \partial \mathcal{D}_N = \overline{\mathcal{D}}_N \backslash \mathcal{D}_N$.
\end{remar}
\noindent
We introduce the positive number $\varepsilon_0 > 0$, which will correspond to the fixed amount of kinetic energy that is lost during any inelastic collision.\\
We define the dynamics of the particles to be the free flow inside $\mathcal{D}_N$.\\
We also say that a single collision takes place when $Z_N \in \overline{\mathcal{D}}_N \backslash \mathcal{D}_N$ and when there exists a single pair $(i,j) \in \mathcal{P}_N$ such that $\vert x_i - x_j \vert = 1$. If in addition $\vert v_i - v_j \vert^2 > 4\varepsilon_0$, an inelastic collision takes place.  
In such a case, the velocities $(v_i,v_j)$ of the pair of colliding particles are immediately changed into $(v_i',v_j')$ according to the law:
\begin{align}
\label{EQUATSect1LoideCollision_}
\text{if}\ \vert v_i - v_j \vert^2 > 4\varepsilon_0, \hspace{5mm} \left\{
\begin{array}{rcl}
v_i' &=& \displaystyle{\frac{v_i + v_j}{2}} - \sigma \sqrt{ \frac{\vert v_j-v_i \vert^2}{4} - \varepsilon_0}, \\
v_j' &=& \displaystyle{\frac{v_i + v_j}{2}} + \sigma \sqrt{ \frac{\vert v_j-v_i \vert^2}{4} - \varepsilon_0},
\end{array}
\right.
\end{align}
where $\sigma$ is the symmetry of the normalized relative velocity $(v_j-v_i)/\vert v_j-v_i \vert$ with respect to $\omega^\perp$, and $\omega$ is the line of contact between the two colliding particles, that is:
\begin{align}
\label{EQUATSect1DefinSigma}
\sigma = \frac{v_j-v_i}{\vert v_j - v_i \vert} - 2 \left( \frac{(v_j-v_i)}{\vert v_j - v_i \vert} \cdot \omega \right) \omega,
\end{align}
and
\begin{align}
\omega = \frac{x_j-x_i}{\vert x_j - x_i \vert} \cdotp
\end{align}
A fixed quantity, equal to $\varepsilon_0 > 0$, of kinetic energy is lost during each inelastic collision:
\begin{align}
\label{EQUATSect1DissipatioEnCin}
\frac{\vert v_i' \vert^2}{2} + \frac{\vert v_j' \vert^2}{2} = \frac{\vert v_i \vert^2}{2} + \frac{\vert v_j \vert^2}{2} - \varepsilon_0.
\end{align}
\noindent
Otherwise, if $\vert v_i - v_j \vert^2 \leq 4\varepsilon_0$, we assume that an elastic collision takes place:
\begin{align}
\label{EQUATSect1LoideColliElast}
\text{if}\ \vert v_i - v_j \vert^2 \leq 4\varepsilon_0, \hspace{5mm} \left\{
\begin{array}{rcl}
v_i' &=& v_i - (v_i-v_j)\cdot\omega\omega, \\
v_j' &=& v_j + (v_i-v_j)\cdot\omega\omega.
\end{array}
\right.
\end{align}

\noindent
The inelastic collision law is obtained as follows. Assuming first that the momentum is conserved during an inelastic collision, but that a fixed amount $\varepsilon_0$ of kinetic energy is lost, the post-collisional velocities $v_i'$ and $v_j'$ are necessarily of the form \eqref{EQUATSect1LoideCollision_}, with an angular parameter $\sigma \in \mathbb{S}^{d-1}$ to be determined. In our model, we assume that the angles formed between the relative velocities and the vector $\omega$ (normal to the plane of contact between the colliding particles), before and after the collision, remain the same up to a change of sign. Besides, a collision does not modify the direction of the orthogonal component with respect to $\omega$ of the relative velocities. This imposes the expression \eqref{EQUATSect1DefinSigma} for $\sigma$. Nevertheless, the magnitude of all the components of the relative velocity are scaled in the same way during the collision. This property makes the model quite different from the classical inelastic collisions described with restitution coefficient, acting only on the normal component of the relative velocity.\\
Finally, observe that by construction such a collision law associates to a pre-collisional relative velocity (that is, such that $(v_j-v_i)\cdot\omega < 0$) a post-collisional relative velocity (that is, such that $(v'_j-v'_i)\cdot\omega > 0$).\\
\newline
We will call the particle system that we introduced above the \emph{inelastic hard sphere with emission model}.\\
From now on, we will refer to this model as the IHSE model, and to its flow in the phase space as the IHSE flow.

\subsection{Definition of the Transport-Collision-Transport (TCT) dynamics}
\label{SSECT_2.2_}

In the previous section, we introduced only formally the dynamics of the IHSE. It remains to define precisely such a dynamics.\\
We will consider a positive time $\tau > 0$, and define the dynamics on such the time interval $[0,\tau]$, assuming that the system experiences at most one collision. We will see that this restriction, at first glance quite strong, is actually enough to define the dynamics of the particle system in general, following the approach originally developed by Alexander in \cite{Alex975}.\\
Under such a restriction, the dynamics we will describe in this section is the composition of three consecutive operations: free transport, collision, and free transport again.\\
Let us mention that a similar approach of decomposing the evolution in terms of transport-collision-transport can be applied to much larger classes of particle systems than the model we are considering.


We will characterize the configurations leading to a collision, grazing or not, using the two following objects.

\begin{defin}[Collision time function, collision times]
Let $d \geq 2$ and $N$ be two positive integers. For any pair $(i,j) \in \mathcal{P}_N$, we define the \emph{collision time function} $g_{i,j}$ as:
\begin{align}
\label{EQUATDefinFonctTempsColli}
g_{i,j}: \mathbb{R}^{dN} \ni W_N = \left(w_1,\dots,w_i,\dots,w_j,\dots,w_N\right)  \mapsto \left\vert w_i-w_j \right\vert^2 -1.
\end{align}
For any configuration $Z_N = (X_N,V_N) \in \mathcal{D}_N$ and any pair $(i,j) \in \mathcal{P}_N$, we define the \emph{virtual collision time of the pair $(i,j)$} as the solution $\tau_{i,j}$ of the equation:
\begin{align}
g_{i,j}\left( X_N + \tau_{i,j} V_N \right) = 0,
\end{align}
if such a solution exists. For any configuration $Z_N = (X_N,V_N) \in \mathcal{D}_N$, we define the \emph{first collision time} $t_c = t_c(Z_N)$ as:
\begin{align}
t_c &= \min_{(i,j) \in \mathcal{P}_N} \left\{\tau_{i,j} > 0 \ /\ g_{i,j}(X_N + \tau_{i,j} V_N) = 0\right\} \hspace{2mm} \text{if there exists a pair } \hspace{0mm}(i,j) \in \mathcal{P}_N \text{such that the equation} \nonumber\\
&\vspace{-3mm}\hspace{80mm}g_{i,j}(X_N+\tau_{i,j}V_N) = 0\hspace{3mm} \text{has a positive solution},\\
t_c &= +\infty \hspace{5mm} \text{otherwise}. \nonumber
\end{align}
that is, $t_c$ is the smallest positive virtual collision time, or again, the smallest positive solution of the equations (in $t$) $g_{i,j}(X_N + t V_N) = 0$, with the minimum taken over all the pairs $(i,j) \in \mathcal{P}_N$.
\end{defin}

\begin{defin}[Grazing collision discriminant]
\label{DEFINGraziColliDiscr}
Let $d \geq 2$ and $N$ be two positive integers. For any configuration $Z_N$ and any pair $(i,j) \in \mathcal{P}_N$, we define the \emph{grazing collision discriminant} as the quantity:
\begin{align}
\Delta_{i,j} = \left((x_i-x_j)\cdot(v_i-v_j)\right)^2 - \vert v_i-v_j \vert^2 \left( \vert x_i-x_j \vert^2 - 1 \right).
\end{align}
\end{defin}

\noindent
The collision time function $g$ introduced in \eqref{EQUATDefinFonctTempsColli} encodes the information that we consider identical, spherical particles, evolving in the whole Euclidean space $\mathbb{R}^d$.

\begin{paragraph}{Definition of the dynamics, without collision.} We start with defining the dynamics of the particles on the time interval $[0,\tau]$, in the case when no collision takes place on such a time interval.

\begin{defin}[Domain of the free dynamics, definition of the IHSE dynamics $\mathcal{T}_\tau$ (I)]
\label{DEFINDomai&TransportLibre}
Let $d \geq 2$ and $N$ be two positive integers and let $\tau > 0$ be a positive real number.\\
We define the \emph{domain of the free dynamics} as the subset of the interior of the phase space $\mathcal{D}_N$, denoted by $\mathcal{D}_N^{(0)}(\tau)$, and defined as:
\begin{align}
\mathcal{D}_N^{(0)}(\tau) = \big\{ Z_N \in \mathcal{D}_N \ /\ t_c(Z_N) > \tau \big\}.
\end{align}
For any $Z_N = (X_N,V_N) \in \mathcal{D}_N^{(0)}(\tau)$, we define the \emph{IHSE dynamics} as the mapping:
\begin{align}
\mathcal{T}_\tau(Z_N) : \mathcal{D}_N^{(0)}(\tau) \ni Z_N \mapsto \mathcal{T}_\tau(Z_N) = \left(X_N + \tau V_N,V_N\right).
\end{align}
\end{defin}

\noindent
By construction, for any $Z_N = (X_N,V_N) \in \mathcal{D}_N^{(0)}(\tau)$ and $0 \leq t \leq \tau$, we have:
\begin{align}
(X_N + t V_N,V_N) \in \mathcal{D}_N,
\end{align}
so that all the trajectories starting from $Z_N \in \mathcal{D}_N^{(0)}(\tau)$ and obtained by free flow can be prolongated on the whole time interval $[0,\tau]$, and in particular no collision takes place on this time interval for such trajectories.
\end{paragraph}

\begin{paragraph}{Definition of the dynamics, with a single collision (TCT dynamics).} 
We define now the dynamics of the particles on the time interval $[0,\tau]$, in the case when a single collision takes place on such a time interval.\\
\newline
We start with some remarks concerning the grazing collisions. Assuming that a collision takes place at $t_c \in ]0,\tau[$, there exists a pair $(i,j) \in \mathcal{P}_N$ such that $g_{i,j}(X_N + t_c V_N) = 0$. Let us consider in addition that:
\begin{align}
\label{EQUATHypotColli_Non_Rasan}
\Delta_{i,j} = \Delta_{i,j}(Z_N) = \left((x_i-x_j)\cdot(v_i-v_j)\right)^2 - \vert v_i-v_j \vert^2 \left( \vert x_i-x_j \vert^2 - 1 \right) > 0.
\end{align}
\noindent
This last condition \eqref{EQUATHypotColli_Non_Rasan} has a simple geometrical meaning. Indeed, let us first observe that $\Delta_{i,j} \geq 0$ is a necessary condition to have the existence of a solution $\tau_{i,j}$ to the equation $g_{i,j}(X_N+\tau_{i,j} V_N) = 0$. If now we assume $\Delta_{i,j} = 0$, the collision time $t_c = \tau_{i,j}$ is given by:
\begin{align}
t_c = - \frac{(x_i-x_j)\cdot(v_i-v_j)}{\vert v_i-v_j \vert^2},
\end{align}
so that at the time of collision we have:
\begin{align}
\label{EQUATDescrColliRasan}
\left[ (x_i+t_c v_i) - (x_j+t_c v_j) \right] \cdot(v_i-v_j) = (x_i-x_j)\cdot(v_i-v_j) + t_c \vert v_i-v_j \vert^2 = 0.
\end{align}
At time $t_c$, the two particles $i$ and $j$ are in contact, and \eqref{EQUATDescrColliRasan} describes the fact that the relative velocity between the two particles is orthogonal to the relative position of the particle $j$ with respect to the particle $i$ at the time of collision. Such a collision is singular, in the sense that the normal component of the relative velocity is zero. Using the terminology used in Definition \ref{DEFINGraziColliDiscr}, we observe a \emph{grazing collision}. In the case of a grazing collision, we cannot give a smooth expression of the time of collision $t_c$ in terms of the initial configuration of the particle system.\\
Let us indeed observe that the condition \eqref{EQUATHypotColli_Non_Rasan} is equivalent to assume that $\partial_t f(t_c,x_i,v_i,x_j,v_j) \neq 0$, where we defined:
\begin{align}
f(t,x,v,y,w) = \left\vert (x+tv) - (y+tw) \right\vert^2 - 1 = \left\vert t(v-w) + (x-y) \right\vert^2 - 1,
\end{align}
since we have:
\begin{align}
\partial_t f(t_c,x_i,v_i,x_j,v_j) = - 2 \frac{\sqrt{\Delta_{i,j}}}{\vert v_i-v_j \vert^2} \cdotp
\end{align}
With such a writing, the collision time $t_c$ satisfies:
\begin{align}
f(t_c,x_i,v_i,x_j,v_j) = 0.
\end{align}
When \eqref{EQUATHypotColli_Non_Rasan} holds, we can apply the implicit function theorem to the function $f$, which provides an expression of the time of collision $t_c$ in terms of the initial configuration, that is, such that $t_c = t_c(x_i,v_i,x_j,v_j)$.\\
\newline


\noindent
In order to be able to define the dynamics of the particle system in the case of one collision, we need assumptions in order to ensure that the trajectory starting from a configuration $Z_N \in \mathcal{D}_N$ presents a single collision in the time interval $]0,\tau[$, and that this collision is involving a single pair of particles.\\
In the following definition, we introduce a subset of initial configurations of $\mathcal{D}_N$ on which we can define trajectories involving a single collision.

\begin{defin}[Domain of the TCT dynamics]
\label{DEFINDomaine_du_TCT_Proce}
Let $d\geq 2$ and $N$ be two positive integers, and let $\varepsilon_0 > 0$ and $\tau > 0$ be two positive real numbers.\\
We define the \emph{domain of the TCT dynamics} as the subset of the interior of the phase space $\mathcal{D}_N$, denoted by $\mathcal{D}^{(1)}_N(\tau)$, and defined as:
\begin{align}
\label{EQUATDefinDomai_TCT_}
\mathcal{D}^{(1)}_N(\tau) = \hspace{-3mm} \bigcup_{(i,j) \in \mathcal{P}_N} \hspace{-3mm} \mathcal{D}^{(1),(i,j)}_N(\tau),
\end{align}
where $\mathcal{D}^{(1),(i,j)}_N(\tau)$ is defined as:
\begin{align}
\mathcal{D}^{(1),(i,j)}_N(\tau) = \bigcap_{p=1}^6 S_p^{(i,j)}
\end{align}
and
\begin{align}
\label{EQUATDefinDmTCTTps_Col_ij_<tau}
S_1^{(i,j)} = \big\{ Z_N \in \mathcal{D}_N\ /\ t_c < \tau \},
\end{align}
\begin{align}
\label{EQUATDefinDmTCTDelta_i_j__>_0_}
S_2^{(i,j)} = \big\{ Z_N \in \mathcal{D}_N\ /\ \Delta_{i,j}(Z_N) > 0 \ \text{and} \ \tau_{i,j} = t_c \big\},
\end{align}
\begin{align}
\label{EQUATDefinDmTCTDelta_k_l__>_0_}
S_3^{(i,j)} = \big\{ Z_N \in \mathcal{D}_N\ /\ \forall\, (k,l) \in \mathcal{P}_N,\, (k,l) \neq (i,j),\, \Delta_{k,l}(Z_N) < 0\ \text{or}\ \big[ \Delta_{k,l}(Z_N) > 0 \ \text{and}\ \tau_{k,l} > \tau \big] \big\},
\end{align}
\begin{align}
\label{EQUATDefinDmTCTEnergCinetNnCrt}
S_4^{(i,j)} = \big\{ Z_N \in \mathcal{D}_N\ /\ \vert v_i - v_j \vert^2 \neq 4\varepsilon_0 \big\},
\end{align}
and
\begin{align}
\label{EQUATDefinDmTCTDelta'i_l__>_0_}
S_5^{(i,j)} = \big\{ Z_N \in \bigcap_{p=1}^4 S_p^{(i,j)} \ /\ \forall\, (k,l) \in \mathcal{P}_N, k=i\ \text{or}\ l=i,\, \Delta_{k,l}(Z_N') < 0 \, \text{or} \, \big[ \Delta_{k,l}(Z_N') > 0 \, \text{and}\ \tau_{k,l}(Z_N') > \tau \big] \big\},
\end{align}
\begin{align}
\label{EQUATDefinDmTCTDelta'k_j__>_0_}
S_6^{(i,j)} = \big\{ Z_N \in \bigcap_{p=1}^4 S_p^{(i,j)} \ /\ \forall\, (k,l) \in \mathcal{P}_N, k=j\ \text{or}\ l=j,\, \Delta_{k,l}(Z_N') < 0 \, \text{or} \, \big[ \Delta_{k,l}(Z_N') > 0 \, \text{and}\ \tau_{k,l}(Z_N') > \tau \big] \big\},
\end{align}
with $Z_N' = (X_N+\tau_{i,j}V_N,V_N')$, and $V_N'$ given by the collision laws \eqref{EQUATSect1LoideCollision_}, \eqref{EQUATSect1LoideColliElast}.

\end{defin}

\noindent
Let us comment on the different assumptions used to define the domain $\mathcal{D}_N^{(1)}(\tau)$ of the TCT dynamics.\\
First, the assumption \eqref{EQUATDefinDmTCTTps_Col_ij_<tau} is necessary to have at least one collision on the time interval $]0,\tau[$, involving a certain pair $(i,j) \in \mathcal{P}_N$. \eqref{EQUATDefinDmTCTDelta_i_j__>_0_} is necessary to have that the first collision involves this pair $(i,j)$, but it is not sufficient. Indeed, another pair $(k,l)$ (or more than another pair) can collide exactly at the same time, and we can also have $i = k$ or $j = l$. The condition \eqref{EQUATDefinDmTCTDelta_k_l__>_0_} prevents such a situation to occur.\\
Therefore, assumptions \eqref{EQUATDefinDmTCTTps_Col_ij_<tau}, \eqref{EQUATDefinDmTCTDelta_i_j__>_0_}, \eqref{EQUATDefinDmTCTDelta_k_l__>_0_} ensure that a first collision takes place before $\tau$, and this collision involves the pair $(i,j)$.\\
When the collision involving the pair $(i,j)$ takes place, one has to make sure that the particles separate after the collision. It is clearly the case if such a collision is elastic (that is, $V_N'$ is computed according to the law \eqref{EQUATSect1LoideColliElast}). In the case of an inelastic collision ($V_N'$ computed according to \eqref{EQUATSect1LoideCollision_}), \eqref{EQUATDefinDmTCTEnergCinetNnCrt} ensures that the post-collisional relative velocity between $i$ and $j$ is not zero.\\
Finally, we want a single collision to take place on the whole time interval $]0,\tau[$. Since at time $t_c$ the velocities of the particles $i$ and $j$ are modified, one has to reexamine the possible collisions in the system, that is now in the configuration $(X_N + t_c V_N,V_N')$. The assumptions \eqref{EQUATDefinDmTCTDelta'i_l__>_0_} and \eqref{EQUATDefinDmTCTDelta'k_j__>_0_} prevent a new collision to take place in time interval $]t_c,\tau[$.\\
All the assumptions of Definition \ref{DEFINDomaine_du_TCT_Proce} ensure that a single collision takes place on the time interval $]0,\tau]$ if $Z_N \in \mathcal{D}^{(1)}_N(\tau)$. In other words, the dynamics of the system for any initial configuration in $\mathcal{D}^{(1)}_N(\tau)$ is given by a Transport-Collision-Transport (TCT) dynamics.

\begin{remar}
The domain $\mathcal{D}^{(1)}_N(\tau)$ is an open set.
\end{remar}
%
\noindent
For $Z_N \in \mathcal{D}^{(1)}_N(\tau)$, no triple collision takes place, and we can define uniquely the post-collisional velocities $v_i'$ and $v_j'$ of the two particles $i$ and $j$ in terms of their pre-collisional velocities $v_i$ and $v_j$ (these pre-collisional velocities corresponding to the initial velocities of these particles), using the reflection laws \eqref{EQUATSect1LoideCollision_}, or \eqref{EQUATSect1LoideColliElast}, depending if the collision is inelastic or not.\\
\newline
We introduce now the four different functions acting on the phase space $\overline{\mathcal{D}}_N$ that we will use in order to define the TCT dynamics. We will define one family of four functions $(\varphi_1,\varphi_2,\varphi_3,\varphi_4)$ for each of the sets $\mathcal{D}_N^{(1),(i,j)}(\tau)$ composing the domain $\mathcal{D}_N^{(1)}(\tau)$ of the TCT dynamics, that is, one family for each pair $(i,j) \in \mathcal{P}_N$.
\begin{defin}[Elementary functions of the TCT dynamics]
\label{DEFIN_TCT_ElemnFonctInelaHSwEm}
Let $d\geq 2$ and $N$ be two positive integers, and let $\varepsilon_0$ and $\tau$ be two positive real numbers.\\
For any $(i,j) \in \mathcal{P}_N$, we define the \emph{elementary functions of the TCT dynamics} $\varphi_k^{(i,j)}$, $k \in \{1,2,3,4\}$, as:
\begin{enumerate}
\item 
Function $\varphi_1^{(i,j)}$ (determination of the collision time $t_c$):
\begin{align}
\label{EQUATDefinFonctVPhi1}
\varphi_1^{(i,j)}: \left\{
\begin{array}{ccc}
\mathcal{D}_N^{(1),(i,j)}(\tau) \subset \mathcal{D}_N &\rightarrow& \mathcal{D}_N\times\, ]0,\tau[ \subset \mathbb{R}^{2dN + 1},\\
Z_N = (X_N,V_N) &\mapsto& \varphi_1^{(i,j)}(Z_N) = (X_N,V_N,t_c),
\end{array}
\right.
\end{align}
with
\begin{align}
\label{EQUATDefinTempsColli}
t_c = \tau_{i,j}(X_N,V_N), \hspace{1mm} \text{the smallest solution } t \text{ of the equation:} \hspace{5mm} g_{i,j}(X_N + t V_N) = 0.
\end{align}
\item Function $\varphi_2^{(i,j)}$ (first free transport, on the time interval $[0,\tau_{i,j}]$):
\begin{align}
\label{EQUATDefinFonctVPhi2}
\varphi_2^{(i,j)}: 
\left\{
\begin{array}{ccc}
\mathcal{D}_N\times\, ]0,\tau[\ \subset \mathbb{R}^{2dN + 1} &\rightarrow& \overline{\mathcal{D}}_N\times\, ]0,\tau[\ \subset \mathbb{R}^{2dN + 1},\\
(X_N,V_N,t) &\mapsto& \varphi_2^{(i,j)}(X_N,V_N,t) = (X_N+tV_N,V_N,t),
\end{array}
\right.
\end{align}
\item Function $\varphi_3^{(i,j)}$ (scattering):
\begin{align}
\label{EQUATDefinFonctVPhi3}
\varphi_3^{(i,j)}: \left\{
\begin{array}{ccc}
\left( \overline{\mathcal{D}}_N \backslash \mathcal{D}_N \cap S_\text{scatt}^{(i,j)} \right) \times\, ]0,\tau[\ \subset \mathbb{R}^{2dN+1} &\rightarrow& \left( \overline{\mathcal{D}}_N \backslash \mathcal{D}_N \right) \times\, ]0,\tau[\ \subset \mathbb{R}^{2dN+1},\\
(X_N,V_N,t) &\mapsto& \varphi_3^{(i,j)}(X_N,V_N,t) = (X_N,V'_N,t),
\end{array}
\right.
\end{align}
with $V'_N = V_N'(X_N,V_N) = (v_1',\dots,v_N')$, such that $v'_k = v_k$ for all $k \neq i,j$, and:
\begin{itemize}
\item if $\vert v_i - v_j \vert^2 > 4\varepsilon_0$, $v'_i,v'_j$ are obtained from an inelastic collision, that is, given by \eqref{EQUATSect1LoideCollision_},
\item if $\vert v_i - v_j \vert^2 < 4\varepsilon_0$, $v'_i$, $v'_j$ are obtained from an elastic collision, that is, given by \eqref{EQUATSect1LoideColliElast},
\end{itemize}
and $S_\text{scatt}^{(i,j)}$ appearing in the domain of $\varphi_3^{(i,j)}$ is defined such as:
\begin{align}
\label{EQUATDefinS_scatt___}
S_\text{scatt}^{(i,j)} = \big\{ Z_N \in \overline{\mathcal{D}}_N\ /\ \vert x_k-x_l \vert > 1, \, \forall\, (k,l) \in \mathcal{P}_N\ /\ (k,l) \neq (i,j)\big\}.
\end{align}
\item Function $\varphi_4^{(i,j)}$ (second free transport, on the time interval $[\tau_{i,j},\tau]$):
\begin{align}
\label{EQUATDefinFonctVPhi4}
\varphi_4^{(i,j)}:
\left\{
\begin{array}{ccc}
\overline{\mathcal{D}}_N \times\, ]0,\tau[\ \subset \mathbb{R}^{2dN+1} &\rightarrow& \overline{\mathcal{D}}_N,\\
(X_N,V_N,t) &\mapsto& \varphi_4^{(i,j)}(X_N,V_N,t) = (X_N + (\tau-t)V_N,V_N).
\end{array}
\right.
\end{align}
\end{enumerate}
\end{defin}

\begin{remar}
The composition of the two first mappings $\varphi_2^{(i,j)} \circ \varphi_1^{(i,j)}$ defines the following transformation:
\begin{align}
(X_N,V_N) \mapsto (X_N + \tau_{i,j} V_N,V_N).
\end{align}
In particular, and although the free transport preserves the measure in the phase space, the mapping defined here corresponds to a transport during a time interval that depends itself on the initial configuration $Z_N$. Therefore, it is not clear a priori how the Lebesgue measure on the phase space evolves under the action of such a flow.
\end{remar}
\noindent
We are now in position to define completely the TCT dynamics, that is, the dynamics of the particle system when a single collision takes place on the time interval $[0,\tau]$.

\begin{defin}[TCT dynamics, definition of the IHSE dynamics $\mathcal{T}_\tau$ (II)]
\label{DEFIN_TCT_ProceInelaHSwEm}
Let $d \geq 2$ and $N$ be two positive integers, and let $\varepsilon_0$ and $\tau$ be two positive real numbers.\\
We define on the domain $\mathcal{D}_N^{(1)}(\tau)$ the \emph{IHSE TCT dynamics} as the mapping:
\begin{align}
\label{EQUATDefin_TCT_Proce}
\mathcal{T}_\tau: \left\{
\begin{array}{ccc}
\mathcal{D}_N^{(1)}(\tau) = \bigcup_{(i,j) \in \mathcal{P}_N} \mathcal{D}_N^{(1),(i,j)}(\tau) &\rightarrow& \mathcal{D}_N \\
Z_N = (X_N,V_N) &\mapsto& \left[ \varphi_4^{(i,j)} \circ \varphi_3^{(i,j)} \circ \varphi_2^{(i,j)} \circ \varphi_1^{(i,j)} \right] (X_N,V_N) \hspace{2mm} \text{if} \hspace{2mm} Z_N \in \mathcal{D}_N^{(1),(i,j)}(\tau),
\end{array}
\right.
\end{align}
where $\varphi_1^{(i,j)}$, $\varphi_2^{(i,j)}$, $\varphi_3^{(i,j)}$ and $\varphi_4^{(i,j)}$ are the elementary functions of the TCT dynamics respectively defined in \eqref{EQUATDefinFonctVPhi1}, \eqref{EQUATDefinFonctVPhi2}, \eqref{EQUATDefinFonctVPhi3} and \eqref{EQUATDefinFonctVPhi4}.
\end{defin}
\end{paragraph}
\noindent
Observe that the definition \eqref{EQUATDefin_TCT_Proce} is consistent, since on the one hand the domain $\mathcal{D}_N^{(1)}(\tau)$ of the TCT dynamics is the disjoint union of the subsets $\mathcal{D}_N^{(1),(i,j)}(\tau)$. On the other hand, by definition of these subsets, the configuration $\mathcal{T}_\tau(Z_N)$ obtained at time $\tau$ is such that no pair of particles are in contact, so the image of $\mathcal{T}_\tau$ is indeed contained in $\mathcal{D}_N$.\\
Let us also observe that, by construction, the image of $\varphi_2^{(i,j)} \circ \varphi_1^{(i,j)}$ is contained in $\overline{\mathcal{D}}_N \backslash \mathcal{D}_N \cap S_\text{scatt}^{(i,j)}$, so that the configuration $\varphi_2^{(i,j)} \circ \varphi_1^{(i,j)}(Z_N)$ does not describe a triple collision, and so the application of the scattering mapping $\varphi_3^{(i,j)}$ is also consistent.
\begin{remar}
The expression \eqref{EQUATDefin_TCT_Proce} for the hard sphere flow can be written in a more direct and explicit way, as follows (assuming that $Z_N$ belongs to $\mathcal{D}_N^{(1),(i,j)}(\tau)$):
\begin{align}
\label{EQUATExpreExpli_Flot}
\mathcal{T}_\tau(X_N,V_N) = \left( X_N + \tau_{i,j} V_N + (\tau-\tau_{i,j})V_N',V_N' \right),
\end{align}
with $\tau_{i,j}$ being the smallest solution of the equation $g_{i,j}(X_N+t V_N) = 0$, and $V_N'$ given according to the reflection laws, that is, either by \eqref{EQUATSect1LoideCollision_}, or by \eqref{EQUATSect1LoideColliElast}.\\
Nevertheless, it will be easier to study the evolution of the Lebesgue measure under the action of the TCT dynamics using the decomposition we introduced.\\
In addition, the decomposition into the four mappings $\varphi_k^{(i,j)}$ enables to treat the evolution of the measure for a class of particle systems which is not restricted to the IHSE model. For instance, the decomposition can be used to study the evolution of the measure in the case of the classical hard sphere flow.\\
Conversely, the short expression \eqref{EQUATExpreExpli_Flot} hides in its compact form an involved dynamics, in the sense that the post-collisional velocity $V_N'$ depends on the point (in the position variable) where the particles collides, such a point being given in terms of the collision time $t_c$, which depends itself on the initial configuration $Z_N$.
\end{remar}

\section{Evolution of the measure of the phase space under the action of the TCT dynamics}
\label{SECTI__3__}

We will study in this section how the Lebesgue measure evolves along the flow of the TCT dynamics. We will perform the computations for general scattering mappings, that is, we will not restrict ourselves to the inelastic hard sphere with emission model. 
More precisely, we will consider more general scattering mappings $\widetilde{\varphi}^{(i,j)}_3$ than $\varphi_3^{(i,j)}$ introduced in Definition \ref{DEFIN_TCT_ElemnFonctInelaHSwEm}. We will then define general TCT dynamics on the elementary domains $\mathcal{D}_N^{(1),(i,j)}(\tau)$ as the composition of the mappings:
\begin{align}
\varphi_4^{(i,j)} \circ \widetilde{\varphi}_3^{(i,j)} \circ \varphi_2^{(i,j)} \circ \varphi_1^{(i,j)}.
\end{align}
With such an approach, we will be able to treat also the classical case of the elastic hard sphere system.\\
We will use continuously the notion of \emph{scattering mapping}, which describes the change of velocities due to collisions.

\begin{defin}[Scattering mapping]
\label{DEFINScatteringMappi}
Let $d \geq 2$ and $N$ be two positive integers. For any $(i,j) \in \mathcal{P}_N$, we consider a mapping $\widetilde{\varphi}^{(i,j)}$ defined as:
\begin{align}
\label{EQUATCondi_TCT_GenerDpdV'}
\widetilde{\varphi}^{(i,j)}: \left\{
\begin{array}{rcl}
\left(\overline{\mathcal{D}}_N \backslash \mathcal{D}_N  \cap S_\text{scatt}^{(i,j)} \right) \times ]0,\tau[ &\rightarrow& \left(\overline{\mathcal{D}}_N \backslash \mathcal{D}_N \right) \times ]0,\tau[,\\
(X_N,V_N,\tau) &\mapsto& (X_N,V'_N,\tau),
\end{array}
\right.
\end{align}
with $S_\text{scatt}^{(i,j)}$ defined in \eqref{EQUATDefinS_scatt___}.\\
We say that $\widetilde{\varphi}$ is a \emph{scattering mapping} if it satisfies the following assumptions:
\begin{align}
V'_N = V'_N(X_N,V_N) \in \mathcal{C}^1(\mathbb{R}^{2dN},\mathbb{R}^{dN}),
\end{align}
$v'_k = v_k$ for all $k \neq i,j$, and for all $(X_N,V_N) \in \mathbb{R}^{dN} \times \mathbb{R}^{dN}$ such that $\vert x_i - x_j \vert = 1$ and $(x_i-x_j)\cdot(v_i-v_j) < 0$, we have:
\begin{align}
(x_i-x_j)\cdot(v'_i-v'_j) > 0,
\end{align}
where $V_N' = (v_1',\dots,v_N')$.
\end{defin}
\noindent
A scattering mapping is acting only on velocities of configurations in the phase space.\\
\newline
\noindent
The evolution of the Lebesgue measure of the phase space under the action of the TCT dynamics is described thanks to the Jacobian determinant of the mapping $\mathcal{T}_\tau$, defined in \eqref{EQUATDefin_TCT_Proce}. We can now state the following result, concerning the evolution of the Lebesgue measure along the flow of a general TCT dynamics.

\begin{theor}[TCT volume formula]
\label{THEOREvoluMesur_TCT_Proce}
Let $\tau > 0$ be a positive real number. Let us consider a general TCT dynamics $\mathcal{T}_\tau$, defined as in Definition \ref{DEFIN_TCT_ProceInelaHSwEm} on the time interval $]0,\tau]$ as the composition of the four functions functions $\varphi_1^{(i,j)}$, $\varphi_2^{(i,j)}$, $\widetilde{\varphi}_3^{(i,j)}$ and $\varphi_4^{(i,j)}$, where the three functions $\varphi_1^{(i,j)}$, $\varphi_2^{(i,j)}$ and $\varphi_4^{(i,j)}$ are defined as in Definition \ref{DEFIN_TCT_ElemnFonctInelaHSwEm}, and for any $(i,j) \in \mathcal{P}_N$, $\widetilde{\varphi}_3^{(i,j)}$ is a scattering mapping in the sense of Definition \ref{DEFINScatteringMappi}.\\
Then, for any $Z_N=(X_N,V_N) \in \mathcal{D}_N^{(1)}(\tau)$, the Jacobian determinant of the flow $\mathcal{T}_\tau$ of the TCT dynamics at the point $Z_N$ is given by:
\begin{align}
\label{EQUATTheo1EvoluMesur_TCT_}
\det\left(\text{Jac}(\mathcal{T}_\tau)\right)(X_N,V_N) &= \left[1 + \nabla_{X_N}t_c \cdot \left(V_N-V'_N\right)\right] \det(N),
\end{align}
where $N$ is the Jacobian matrix of the post-collisional velocity $V'_N = (v'_1,\dots,v'_N)$ with respect to the velocity variables $V_N=(v_1,\dots,v_N)$ at the time of collision $t_c(X_N,V_N)$ (defined in \eqref{EQUATDefinTempsColli}). In other words, for any $(i,j) \in \mathcal{P}_N$ and $1 \leq j,l \leq d$:
\begin{align}
N_{(i,j),(k,l)} = \partial_{v_{k,l}} v'_{i,j}(X_N+t_cV_N,V_N),
\end{align}
where $v_{k,l}$ is the $l$-th component of the vector $v_k \in \mathbb{R}^d$, the $k$-th velocity of the initial configuration $Z_N=(X_N,V_N)$, and $v'_{i,j}$ is the $j$-th component of the $i$-th post-collisional velocity.
\end{theor}

\begin{proof}
We consider an initial configuration $Z_N = (X_N,V_N)$ belonging to the domain $\mathcal{D}_N^{(1)}(\tau)$ of the TCT dynamics. Let us start with computing the respective Jacobian matrices of the four elementary mappings $\varphi_i$, defined respectively in \eqref{EQUATDefinFonctVPhi1}, \eqref{EQUATDefinFonctVPhi2}, \eqref{EQUATDefinFonctVPhi3} and \eqref{EQUATDefinFonctVPhi4}. We will denote by $M_i$ the Jacobian matrix of $\varphi_i$.\\
We find:
\begin{align}
M_1(X_N,V_N) = \text{Jac}\left(\varphi_1\right)(X_N,V_N) = \begin{pmatrix}
I_{dN} & 0 \\
0 & I_{dN} \\
^t\hspace{-1mm}\left(\nabla_{X_N} t_c\right) & ^t\hspace{-1mm}\left(\nabla_{V_N} t_c\right)
\end{pmatrix} \in \mathcal{M}_{(2dN+1) \times 2dN}(\mathbb{R}),
\end{align}
where $I_{dN}$ denotes the $dN \times dN$ identity matrix, $\nabla_{X_N} t_c$ and $\nabla_{V_N} t_c$ the gradients of the collision time $t_c = t_c(X_N,V_N)$ (defined in \eqref{EQUATDefinTempsColli}), respectively, with respect to the position and the velocity variables, that is, we have:
\begin{align}
^t\hspace{-1mm}\left(\nabla_X t_c\right) = \begin{pmatrix} \partial_{x_{1,1}} t_c & \dots & \partial_{x_{1,d}} t_c & \dots & \partial_{x_{i,1}} t_c & \dots & \partial_{x_{i,d}} t_c & \dots \end{pmatrix} \in \mathcal{M}_{1 \times dN}(\mathbb{R}).
\end{align}
The Jacobian matrix of $\varphi_2$ is:
\begin{align}
M_2(X_N,V_N,t) = \text{Jac}(X_N,V_N,t)\left(\varphi_2\right) = \begin{pmatrix}
I_{dN} & t I_{dN} & V_N \\
0 & I_{dN} & 0 \\
0 & 0 & 1
\end{pmatrix} \in \mathcal{M}_{(2dN+1)\times(2dN+1)}(\mathbb{R}).
\end{align}
Finally, the respective Jacobian matrices $M_3$ and $M_4$ of $\varphi_3$ and $\varphi_4$ are:
\begin{align}
M_3(X_N,V_N,t) = \text{Jac}(\varphi_3)(X_N,V_N,t) = \begin{pmatrix}
I_{dN} & 0 & 0 \\
N_1(X_N,V_N) & N_2(X_N,V_N) & 0 \\
0 & 0 & 1
\end{pmatrix} \in \mathcal{M}_{(2dN+1)\times(2dN+1)}(\mathbb{R}),
\end{align}
and
\begin{align}
M_4(X_N,V_N,t) = \text{Jac}(\varphi_4)(X_N,V_N,t) = \begin{pmatrix}
I_{dN} & (\tau-t) I_{dN} & -V_N \\
0 & I_{dN} & 0
\end{pmatrix}  \in \mathcal{M}_{2dN\times(2dN+1)},
\end{align}
with $N_1$, $N_2$ the respective Jacobian matrices of the post-collisional velocities $V_N' = V_N'(X_N,V_N)$, respectively, with respect to the position and velocity variables. In other words, we have, for $N_2$:
\begin{align}
N_2 = \begin{pmatrix} \partial_{V_q} \left(V'_N(X_N,V_N)\right)_p \end{pmatrix}_{p,q},
\end{align}
with $p,q \in dN$, both indices describing which particle ($1 \leq i \leq N$) and which coordinate ($1 \leq k \leq d$) are considered (for example, with $q = d(i-1) + k$, we have $\partial_{V_q} = \partial_{v_{i,k}}$).\\
Therefore, we find for the Jacobian matrix of the TCT dynamics that is defined in \eqref{EQUATDefin_TCT_Proce}:
\begin{align}
\label{EQUATExpreMatriJacob_TCT_}
\text{Jac}(\mathcal{T}_\tau)(X_N,V_N) &= M_4(X_N+t_cV_N,V'_N,t_c) \cdot M_3(X_N+t_cV_N,V_N,t_c) \cdot M_2(X_N,V_N,t_c) \cdot M_1(X_N,V_N) \nonumber\\
&= \begin{pmatrix} P_1 & P_2 \\
N_1 + N_1 \cdot(\nabla_{X_N}t_c \otimes V_N) & t_c N_1 + N_1 \cdot (\nabla_{V_N} t_c \otimes V_N) + N_2
\end{pmatrix}
\end{align}
with
\begin{align}
P_1 = I_{dN}+\nabla_{X_N}t_c \otimes \left(V_N-V'_N\right) + (\tau-t_c) N_1 + (\tau-t_c) N_1 \cdot (\nabla_{X_N}t_c\otimes V_N),
\end{align}
and
\begin{align}
P_2 = t_c I_{dN} + \nabla_{V_N}t_c\otimes(V_N-V'_N) + t_c(\tau-t_c) N_1 + (\tau-t_c)N_1 \cdot (\nabla_{V_N}t_c\otimes V_N) + (\tau-t_c)N_2,
\end{align}
the matrices $N_1$ and $N_2$ being evaluated at $(X_N+t_cV_N,V_N)$.\\
In order to be able to compute the determinant of the matrix $\text{Jac}(\mathcal{T}_\tau)(X_N,V_N)$, we need to be able to compare the gradients $\nabla_{X_N}t_c$ and $\nabla_{V_N}t_c$. Fixing an index $1 \leq k \leq N$, corresponding to the label of one particle, and fixing one coordinate $1 \leq l \leq d$, we find on the one hand:
\begin{align}
0 = \partial_{x_{k,l}} \left[(X_N,V_N) \mapsto g_{i,j}(X_N+t_cV_N)\right] = \partial_{w_{k,l}} g_{i,j}(X_N+t_cV_N) + \partial_{x_{k,l}}t_c \left[ V_N \cdot \nabla g_{i,j}(X_N+t_cV_N) \right],
\end{align}
and on the other hand:
\begin{align}
0 = \partial_{v_{k,l}} \left[(X_N,V_N) \mapsto g_{i,j}(X_N+t_cV_N)\right] = \partial_{v_{k,l}}t_c \left[ V_N \cdot \nabla g_{i,j}(X_N+t_cV_N) \right] + t_c \partial_{w_{k,l}} g_{i,j}(X_N+t_cV_N).
\end{align}
We observe now that $V_N \cdot \nabla g_{i,j}(X_N+t_cV_N) \neq 0$. Therefore, we obtain the key equation:
\begin{align}
\partial_{v_{k,l}}t_c = - t_c \frac{\partial_{w_{k,l}}g_{i,j}(X_N+t_cV_N)}{V_N \cdot \nabla g_{i,j}(X_N+t_cV_N)} = t_c \left[ \partial_{x_{k,l}}t_c \right],
\end{align}
which can be rewritten as:
\begin{align}
\label{EQUATLien_entreNablXNablV_t_c_}
\nabla_{V_N}t_c = t_c \left[ \nabla_{X_N}t_c\right].
\end{align}
The relation \eqref{EQUATLien_entreNablXNablV_t_c_} implies substantial simplifications when computing the Jacobian determinant of the matrix \eqref{EQUATExpreMatriJacob_TCT_}. Indeed, using the multilinearity and the alternating property of the determinant, we find first, by substracting $t_c$ times the first $dN$ columns to the $dN$ last ones:
\begin{align}
\det\left(\text{Jac}(\mathcal{T}_\tau)\right)(X_N,V_N) = \begin{vmatrix}
P_1 & (\tau-t_c)N_2 \\
N_1 + N_1 \cdot \left( \nabla_{X_N}t_c \otimes V_N \right) & N_2
\end{vmatrix}.
\end{align}
Then, substracting $(\tau-t_c)$ times the last $dN$ rows to the first $dN$ ones, we obtain:
\begin{align}
\det\left(\text{Jac}(\mathcal{T}_\tau)\right)(X_N,V_N) =
\begin{vmatrix}
I_{dN} + \nabla_{X_N} t_c \otimes (V_N-V'_N) & 0 \\
N_1 + N_1 \cdot \left( \nabla_{X_N}t_c \otimes V_N \right) & N_2
\end{vmatrix}.
\end{align}
We can now develop the determinant, providing:
\begin{align}
\det\left(\text{Jac}(\mathcal{T}_\tau)\right)(X_N,V_N) &= \det\left( I_{dN} + \nabla_{X_N} t_c \otimes (V_N-V'_N) \right) \det(N_2) \nonumber\\
&= \left[1 + \nabla_{X_N}t_c \cdot \left(V_N-V'_N\right)\right] \det(N_2),
\end{align}
and the proof of Theorem \ref{THEOREvoluMesur_TCT_Proce} is complete.
\end{proof}

\begin{remar}
Let us observe that the matrix $N_1$ (the Jacobian matrix of the post-collisional velocity $V'_N$ with respect to the position variables) plays no role in the final result.
\end{remar}

\noindent
We finally observe that it would have been possible to consider even more general TCT dynamics, assuming only that the collision time $t_c$ is given by an equation of the form $g(X_N+t_cV_N) = 0$ (the condition that defines the time of collision), with a general collision time function $g$ that satisfies only:
\begin{align}
\label{EQUATCondi_TCT_GenerNnGrz}
V_N \cdot \nabla g(X_N+t_cV_N) \neq 0.
\end{align}
\noindent
In the case of the TCT dynamics introduced in Definition \ref{DEFIN_TCT_ProceInelaHSwEm}, the function $g$ is defined as follows:
\begin{align}
g: \mathcal{D}^{(1)}_N(\tau) \rightarrow \mathbb{R},\hspace{3mm} Z_N \mapsto g_{i,j}(Z_N) \hspace{1mm} \text{if} \hspace{1mm} Z_N \in \mathcal{D}_N^{(1),(i,j)},
\end{align}
where the domain $\mathcal{D}^{(1)}_N(\tau)$ is defined in \eqref{EQUATDefinDomai_TCT_}. In this case, the condition \eqref{EQUATCondi_TCT_GenerNnGrz} is fulfilled by construction, since it corresponds to \eqref{EQUATHypotColli_Non_Rasan}, which holds true on the TCT domain $\mathcal{D}^{(1)}_N(\tau)$.\\
\newline
The condition $V_N \cdot \nabla g(X_N+t_cV_N) \neq 0$ generalizes the condition \eqref{EQUATHypotColli_Non_Rasan}, prescribing that the collision is not grazing.\\
In particular, in the case of a single particle colliding against an obstacle $\Omega = \{ x \in\mathbb{R}^d\ /\ g(x) < 0 \}$ (so that the boundary of the obstacle is given by the level set $\{x \in \mathbb{R}^d \ /\ g(x) = 0\}$ of the collision time function $g$), the condition that a collision takes place can be written as $g(x+t_cv)=0$, and the condition $v \cdot \nabla g(x+t_cv) \neq 0$ means that the velocity $v$ of the particle is not in the tangent plane to the obstacle at the point of collision $x+tv$.\\
The description in terms of an obstacle is also consistent with the approach consisting of seeing a system of $N$ hard spheres evolving in $\mathbb{R}^d$ as a single particle evolving in dimension $dN$ as in a billiard, the boundary of such a billiard being the union of the cylinders $\vert x_i-x_j \vert = 1$, corresponding to the condition of non-overlapping of the particles.

\section{Applications of the TCT volume formula \eqref{EQUATTheo1EvoluMesur_TCT_}}
\label{SECTI__4__}

In this section we will apply the general result of Theorem \ref{THEOREvoluMesur_TCT_Proce} to particular cases. The classical proof of Alexander (\cite{Alex975}) exhibits that, in order to show the well-posedness of the classical, elastic hard sphere particle system, it is enough to show that the volume is preserved in the phase space, under the action of a single TCT dynamics (as introduced in Definition \ref{DEFIN_TCT_ProceInelaHSwEm}) on the time interval $]0,\tau]$, where $\tau$ is any positive time, fixed a priori.\\
In a first time, we will apply \eqref{EQUATTheo1EvoluMesur_TCT_} to the elastic hard sphere system, recovering the well-known result that the volume is conserved in the phase space for such systems. In a second time, we will turn to the IHSE model, which is the central focus of the present article.

\subsection{The TCT volume formula, in the case of elastic hard spheres}

Concerning the classical elastic hard sphere system, as a consequence of Theorem \ref{THEOREvoluMesur_TCT_Proce}, we recover the following classical result.

\begin{theor}[TCT volume formula, elastic hard sphere case]
\label{THEOREvoluMesurTCTHSProce}
Let us consider the TCT dynamics corresponding to the elastic hard sphere system.\\
In other words, the TCT dynamics is defined as in Definition \ref{DEFIN_TCT_ProceInelaHSwEm}, except that the definition of the scattering mapping $\varphi_3^{(i,j)}$ is modified such that only elastic collisions (where $v_i'$ and $v_j'$ are given by \eqref{EQUATSect1LoideColliElast}) take place (the definition of the three other mappings $\varphi_1^{(i,j)}$, $\varphi_2^{(i,j)}$ and $\varphi_4^{(i,j)}$ are unchanged from Definition \ref{DEFIN_TCT_ElemnFonctInelaHSwEm}). Its associated domain $\mathcal{D}^{(1)}_{N,\text{HS}}(\tau)$ is defined as in Definition \ref{DEFINDomaine_du_TCT_Proce}, except that the assumption \eqref{EQUATDefinDmTCTEnergCinetNnCrt} is omitted. $g_{i,j}$ is defined by \eqref{EQUATDefinFonctTempsColli} for any $(i,j) \in \mathcal{P}_N$.\\
Then, for any $Z_N=(X_N,V_N) \in \mathcal{D}^{(1)}_{N,\text{HS}}(\tau)$, and for any neighbourhood of $Z_N$ in $\mathcal{D}^{(1)}_{N,\text{HS}}(\tau)$, the TCT dynamics of the elastic hard sphere system preserves locally the Lebesgue measure in the phase space. In other words, we have:
\begin{align}
\label{EQUATTheo2EvoluMesurElaHS}
\left\vert \det\left(\text{Jac}(\mathcal{T}_\tau)\right)(X_N,V_N) \right\vert &= 1.
\end{align}
\end{theor}

\begin{proof}
Let us recall that we denote the unit vector $(x_i+t_c v_i) - (x_j+t_c v_j) = (x_i-x_j) + t_c(v_i-v_j)$ by $\omega$.\\
The result is obtained by applying the formula \eqref{EQUATTheo1EvoluMesur_TCT_}. More precisely, relying on the explicit expression of the collision time function $g_{i,j}$ we find for $\nabla_{X_N} t_c$:
\begin{align}
\nabla_{x_i} t_c = - \frac{2 \omega}{V_N\cdot\nabla g\left(X_N+t_cV_N\right)} = - \frac{\omega}{(v_i-v_j)\cdot \omega} \hspace{5mm} \text{and} \hspace{5mm} \nabla_{x_j} t_c = \frac{\omega}{(v_i-v_j)\cdot\omega} \cdotp
\end{align}
Therefore, we find:
\begin{align}
\label{EQUATTCTHSPrdScNbTVV}
\nabla_{X_N}t_c \cdot(V_N-V'_N) &= \nabla_{x_i}t_c\cdot(v_i-v'_i) + \nabla_{x_j}t_c\cdot(v_j-v'_j) \nonumber\\
&= - \frac{\omega}{(v_i-v_j) \cdot \omega}\cdot\left[(v_i-v_j)\cdot\omega \omega\right] + \frac{\omega}{(v_i-v_j)\cdot\omega} \cdot \left[ -(v_i-v_j)\cdot\omega \omega\right] = -2.
\end{align}
On the other hand, the mapping that associates $(v_i',v_j')$ is linear in $(v_i,v_j)$, therefore we have:
\begin{align}
V_N' = N V_N,
\end{align}
where the matrix $N$ is the Jacobian matrix of the post-collisional velocity $V'_N$ with respect to the variable $V_N$, which appears in the formula \eqref{EQUATTheo1EvoluMesur_TCT_}. Since at $X_N$ fixed the mapping $V_N \mapsto V'_N$ is an involution, we have:
\begin{align}
\label{EQUATTCTHSDeterMatri__N__}
\det(N) = \pm 1.
\end{align}
Inserting \eqref{EQUATTCTHSPrdScNbTVV} and \eqref{EQUATTCTHSDeterMatri__N__} in the formula \eqref{EQUATTheo1EvoluMesur_TCT_}, we deduce the result of Theorem \ref{THEOREvoluMesurTCTHSProce}.
\end{proof}

\subsection{The TCT volume formula, in the case of inelastic hard spheres with emission}

We now turn to the main model of the present article: the inelastic hard sphere with emission system. We will see that the application of formula \eqref{EQUATTheo1EvoluMesur_TCT_} to this model has important consequences.\\
Before performing such a computation, we emphasize that the flow we introduced is not injective. Therefore, using the word ``flow'' is slightly abusive, but the dynamics of the particles being deterministic (assuming that it is well-defined), considering the evolution of the measure under the action of such a particle dynamics remains meaningful. The lack of injectivity is due to the following phenomenon: if an inelastic collision takes place between two particles, and if the energy of these two particles involved in such a collision is smaller than $\varepsilon_0$ \emph{immediately after} the collision, the configuration of the system can also be reached from an elastic collision between the same pair of particles.\\
However, if we consider the flow of the particle system restricted to the subsets of the phase space for which this flow is injective, it is possible to estimate the evolution of the Lebesgue measure. In order to have a locally injective flow, one can consider for example, for a given initial configuration of the particle system, a neighbourhood of this configuration that is small enough to ensure that all the other configurations of this neighbourhood lead to similar trajectories, in the sense that the number of the consecutive collisions, and their respective types (elastic or inelastic, and involving the same pairs of particles) are the same.\\
The evolution of the Lebesgue measure under the action of the IHSE flow will be deduced from the TCT volume formula \eqref{EQUATTheo1EvoluMesur_TCT_}, applied to the IHSE model. We start with rewriting in terms of the velocities of the colliding particles the term $1 + \nabla_{X_N}t_c \cdot (V_N-V'_N)$ in the right hand side of \eqref{EQUATTheo1EvoluMesur_TCT_}.

\begin{propo}
\label{PROPOPrdSctcVV'TCTProceIn}
Let us consider the TCT dynamics corresponding to the IHSE system, introduced in Definition \ref{DEFIN_TCT_ProceInelaHSwEm}, associated to the open domain $\mathcal{D}^{(1)}_N(\tau)$, defined as in Definition \ref{DEFINDomaine_du_TCT_Proce}, with $g_{i,j}$ being defined by \eqref{EQUATDefinFonctTempsColli}.\\
Then, for any $Z_N=(X_N,V_N) \in \mathcal{D}^{(1)}_N(\tau)$ such that
\begin{align}
\label{EQUATHypotPrdSctcVV'EnrgC}
\vert v_i - v_j \vert^2 > 4\varepsilon_0,
\end{align}
and for any neighbourhood of $Z_N$ in $\mathcal{D}^{(1)}_N(\tau)$ small enough, we have:
\begin{align}
\label{EQUATPrdSctcVV'TCTProceIn}
1 + \nabla_{X_N}t_c \cdot (V_N-V'_N) = - \sqrt{1 - \frac{4\varepsilon_0}{\vert v_i - v_j \vert^2}}.
\end{align}
\end{propo}

\begin{proof}
The computation of the scalar product $\nabla_{X_N}t_c \cdot(V_N-V'_N)$ is very similar to the elastic hard sphere case addressed in Theorem \ref{THEOREvoluMesurTCTHSProce}. In particular, since the function $g_{i,j}$, defining the collision time $t_c$, is the same, the only change lies in the difference $V_N-V'_N$. Relying on \eqref{EQUATSect1LoideCollision_} and \eqref{EQUATSect1DefinSigma}, for the respective definitions of $(v'_i,v'_j)$ and $\sigma$, we find:
\begin{align}
\nabla_{X_N}t_c \cdot (V_N-V'_N) &= \nabla_{x_i}t_c \cdot (v_i-v'_i) + \nabla_{x_j}t_c \cdot (v_j-v'_j) \nonumber\\
&= 2 \left[ -\frac{(v_i-v_j)\cdot\omega}{2(v_i-v_j)\cdot\omega} - \frac{\sigma \cdot \omega}{(v_i-v_j)\cdot\omega} \sqrt{\frac{\vert v_i-v_j \vert^2}{4} - \varepsilon_0} \right] \nonumber\\
&= -1 - \frac{2(v_i-v_j)\cdot\omega}{\vert v_i-v_j \vert (v_i-v_j)\cdot\omega} \sqrt{\frac{\vert v_i-v_j \vert^2}{4} - \varepsilon_0} \nonumber\\
&= -1 - \sqrt{1 - \frac{4\varepsilon_0}{\vert v_i-v_j \vert^2}},
\end{align}
and the result of Proposition \ref{PROPOPrdSctcVV'TCTProceIn} follows.
\end{proof}


\noindent
Let us now turn to the second term of the formula \eqref{EQUATTheo1EvoluMesur_TCT_}, that is, the Jacobian determinant of the matrix $N$. Here, we have to study the scattering mapping, defined by \eqref{EQUATSect1LoideCollision_} in case of inelastic collisions. Such a mapping acts only on the velocity variable $V_N$ in the phase space, and formula \eqref{EQUATTheo1EvoluMesur_TCT_} implies in particular that only the partial derivatives of the scattering mapping with respect to the velocity variables matter.\\
Contrary to the elastic hard spheres, the scattering for the IHSE is not given by a linear mapping in $V_N$. However, we will prove a surprising result concerning such a scattering: although a positive amount of kinetic energy is lost in any collision that is energetic enough, such a scattering will preserve the measure in the phase space.\\
We will make use of the following notation throughout the rest of this section.

\begin{notat}[Tensor product of two vectors]
For two vectors $u$ and $v$ that belong to $\mathbb{R}^d$, we denote by $u\otimes v$ the matrix of $\mathcal{M}_{d\times d}(\mathbb{R})$ associated to the following linear mapping:
\begin{align}
\label{EQUATSect1NotatProduTenso}
\mathbb{R}^d \ni x \mapsto (v\cdot x)u,
\end{align}
where $(v\cdot x)$ denotes the scalar product of the two vectors $v$ and $x$.
\end{notat}
\noindent
We study now in detail the scattering \eqref{EQUATSect1LoideCollision_}. In particular, we have the following result.

\begin{theor}[Measure-preserving property of the scattering of the inelastic hard sphere with emission model]
\label{THEORSect2ConservatiMesur}
In dimension $d=2$, the scattering mapping of the inelastic hard sphere with emission model (\eqref{EQUATSect1LoideCollision_} if $\vert v_i - v_j \vert^2 > 4\varepsilon_0$, \eqref{EQUATSect1LoideColliElast} else) preserves locally the measure in the phase space.\\
In other words, at $t$ and $X_N$ fixed, the absolute value of the Jacobian determinant of the scattering $V_N \mapsto V'_N = V'_N(X_N,V_N)$, associated to the inelastic hard spheres with emission, is equal to $1$.
\end{theor}

\begin{remar}
In terms of the notations introduced in Theorem \ref{THEOREvoluMesur_TCT_Proce}, Theorem \ref{THEORSect2ConservatiMesur} states that $\vert \det(N) \vert = 1$. We introduced a particle system such that its scattering, surprisingly, does not always conserve the kinetic energy, but does always conserve locally the measure in the phase space.\\
Let us emphasize though that the result holds only for the dimension $d=2$.
\end{remar}
\noindent
In order to prove Theorem \ref{THEORSect2ConservatiMesur}, we will a generalization of the well-know formula:
\begin{align}
\det(I_d + u \otimes \omega) = 1 + u\cdot\omega,
\end{align}
concerning the determinant of a single tensor product, and which is a particular case of the Sylvester's determinant theorem, stating that $\det(I_m + AB) = \det(I_n + BA)$ for $A$ a $m\times n$ matrix, and $B$ a $n \times m$ matrix.\\
In the two-dimensional case we have the following result:
\begin{lemma}[Determinant of the sum of tensors products of two vectors]
We consider the two-dimensional case: $d=2$. Let $\lambda$, $\mu$, $\nu$ be three real numbers, and $u$, $\omega$ be two vectors of $\mathbb{R}^2$.\\
Then:
\begin{align}
\label{EQUATLemmeDeterProduTenso}
\det(I_2 + \lambda u\otimes u + \mu u\otimes\omega + \nu \omega\otimes\omega) = 1 + \lambda \vert u \vert^2 + \mu u\cdot \omega + \nu \vert \omega \vert^2 + \lambda \nu \left(\det(u,\omega)\right)^2.
\end{align}
\end{lemma}

\begin{proof}[Proof of Theorem \ref{THEORSect2ConservatiMesur}]
In the case when the scattering corresponds to an elastic collision, it is well-known (and we have shown it along the proof of Theorem \ref{THEOREvoluMesurTCTHSProce}) that such a mapping is a linear involution, therefore it conserves the measure.\\
Let us then assume that an inelastic collision takes place, described by the scattering mapping \eqref{EQUATSect1LoideCollision_}. To simplify the notations, let us denote by $v$ and $v_*$ the pre-collisional velocities of the two colliding particles $i$ and $j$, and by $v'$ and $v'_*$ their post-collisional velocities. If we denote:
\begin{align}
\kappa = \sqrt{ \frac{\vert v_* - v \vert^2}{4} - \varepsilon_0 },
\end{align}
we find, for any pair of components $i,j \in \{1,\dots,d\}$:
\begin{align}
\label{EQUATSect2DerivPartiSigma}
\partial_{v_i} \sigma_k =  - \frac{\delta_{i,k}}{\vert v_*-v \vert} + \frac{(v_{*,k}-v_k)(v_{*,i}-v_i)}{\vert v_*-v \vert^3} + 2 \frac{\omega_i\omega_k}{\vert v_*-v \vert} - 2 (v_*-v)\cdot\omega\frac{(v_{*,i}-v_i)\omega_k}{\vert v_*-v \vert^3}
\end{align}
where $\sigma_k$ is the $k$-th component of the vector $\sigma$ defined in \eqref{EQUATSect1DefinSigma}, $\delta_{i,k}$ is the Kronecker's delta symbol, equal to $1$ if and only if $i=k$, and $0$ in the other cases. We find also:
\begin{align}
\label{EQUATSect2DerivPartiKappa}
\partial_{v_i} \kappa = \partial_{v_i} \sqrt{ \frac{\vert v_*-v \vert^2}{4} - \varepsilon_0 } = - \frac{(v_{*,i}-v_i)}{4 \kappa} \cdotp
\end{align}
\noindent
Therefore, since the expression of the post-collisional velocities $v'$ and $v'_*$ are given by \eqref{EQUATSect1LoideCollision_}, collecting \eqref{EQUATSect2DerivPartiSigma} and \eqref{EQUATSect2DerivPartiKappa} together we obtain:
\begin{align}
\label{EQUATSect2DerivPartiVPost}
\nabla_v v' &= \frac{1}{2} I_d - \frac{\kappa}{\vert v*-v \vert} I_d + \frac{\kappa}{\vert v_*-v \vert} \frac{(v_*-v)}{\vert v_*-v \vert} \otimes \frac{(v_*-v)}{\vert v_*-v \vert} + 2 \frac{\kappa}{\vert v_*-v \vert} \omega \otimes \omega \nonumber \\
&\hspace{5mm} - 2 \frac{\kappa}{\vert v_*-v \vert} \frac{(v_*-v)\cdot\omega}{\vert v_*-v \vert} \frac{(v_*-v)}{\vert v_*-v \vert} \otimes \omega + \frac{(v_*-v)\otimes(v_*-v)}{4 \vert v_*-v \vert \kappa} - 2\frac{(v_*-v)\cdot\omega}{\vert v_*-v \vert} \frac{(v_*-v)\otimes\omega}{4 \kappa} \nonumber\\
&= \frac{1}{2} I_d + \frac{\kappa}{\vert v_*-v \vert} \Big[ - I_d + \left( 1 + \frac{\vert v_*-v \vert^2}{\kappa^2} \right) \frac{(v_*-v)}{\vert v_*-v \vert} \otimes \frac{(v_*-v)}{\vert v_*-v \vert} \nonumber\\
&\hspace{35mm}- 2 \frac{(v_*-v)}{\vert v_*-v \vert}\cdot\omega \left( 1+\frac{\vert v_*-v \vert^2}{4\kappa^2} \right) \frac{(v_*-v)}{\vert v_*-v \vert} \otimes \omega + 2 \omega\otimes\omega \Big] 
\end{align}
where we used the notation introduced in \eqref{EQUATSect1NotatProduTenso} for the tensor product of two vectors, and $I_d$ is the $d \times d$ identity matrix. Writing the formula \eqref{EQUATSect2DerivPartiVPost} in the form:
\begin{align*}
\nabla_v v' = \frac{1}{2} I_d + A
\end{align*}
with
\begin{align*}
A &= \frac{\kappa}{\vert v_*-v \vert} \Big[ - I_d + \left( 1 + \frac{\vert v_*-v \vert^2}{\kappa^2} \right) \frac{(v_*-v)}{\vert v_*-v \vert} \otimes \frac{(v_*-v)}{\vert v_*-v \vert} \nonumber\\
&\hspace{55mm}- 2 \frac{(v_*-v)}{\vert v_*-v \vert}\cdot\omega \left( 1+\frac{\vert v_*-v \vert^2}{4\kappa^2} \right) \frac{(v_*-v)}{\vert v_*-v \vert} \otimes \omega + 2 \omega\otimes\omega \Big],
\end{align*}
and computing the other partial derivatives $\nabla_{v_*} v'$, $\nabla_v v_*'$ and $\nabla_{v_*} v_*'$, we obtain an expression of the following form for the Jacobian matrix $J$ of the scattering mapping $(v,v_*) \mapsto (v',v_*')$ defined by \eqref{EQUATSect1LoideCollision_}:
\begin{align}
J = \begin{pmatrix} \nabla_{v} v' & \nabla_{v} v'_* \\ \nabla_{v_*} v' & \nabla_{v_*} v_*' \end{pmatrix} = \begin{pmatrix} \frac{1}{2}I_d + A & \frac{1}{2}I_d - A \\ \frac{1}{2}I_d - A & \frac{1}{2}I_d + A \end{pmatrix}.
\end{align}
The determinant of such a matrix can be computed as follows. First we obtain:
\begin{align*}
\det(J) = \begin{vmatrix} \frac{1}{2}I_d + A & \frac{1}{2}I_d - A \\ \frac{1}{2}I_d - A & \frac{1}{2}I_d + A \end{vmatrix} = \begin{vmatrix} \frac{1}{2} I_d + A & \frac{1}{2}I_d - A \\ I_d & I_d \end{vmatrix} = \begin{vmatrix} 2A & \frac{1}{2}I_d - A \\ 0 & I_d \end{vmatrix} = \det(2A).
\end{align*}
It remains to compute $\det(2A)$. Such a determinant is of the form:
\begin{align}
\label{EQUATSect2PrsntDeterminan}
\det(2A) &= \left( -\frac{2\kappa}{\vert v_*-v \vert}\right)^d \det(I_d + \lambda u\otimes u + \mu u\otimes\omega + \nu \omega\otimes \omega),
\end{align}
where
\begin{align}
u = \frac{(v_*-v)}{\vert v_*-v \vert},
\end{align}
\begin{align}
\left\{
\begin{array}{rcl}
\lambda &=& - \left( 1 + \frac{\vert v_*-v \vert^2}{\kappa^2} \right),\\
\mu &=& 2 \frac{(v_*-v)}{\vert v_*-v \vert}\cdot\omega \left( 1+\frac{\vert v_*-v \vert^2}{4\kappa^2} \right),\\
\nu &=& -2.
\end{array}
\right.
\end{align}
Applying now the formula \eqref{EQUATLemmeDeterProduTenso}, \eqref{EQUATSect2PrsntDeterminan} becomes:
\begin{align}
\det(2A) &= \frac{4 \kappa^2}{\vert v_*-v \vert^2} \Bigg[ 1 - \left( 1+\frac{\vert v_*-v \vert}{4\kappa^2} \right) + 2 \frac{(v_*-v)}{\vert v_*-v \vert}\cdot \omega \left(1+\frac{\vert v_*-v \vert^2}{4\kappa^2}\right) \left(\frac{(v_*-v)}{\vert v_*-v \vert}\cdot\omega\right) \nonumber\\
&\hspace{70mm} -2 + 2 \left(1+\frac{\vert v_*-v \vert^2}{\kappa^2}\right) \det\left( \frac{(v_*-v)}{\vert v_*-v \vert},\omega \right)^2 \Bigg] \nonumber\\
&= \frac{4 \kappa^2}{\vert v_*-v \vert^2} \Bigg[ -1 - \left( 1+\frac{\vert v_*-v \vert}{4\kappa^2} \right) + 2 \left(1-\frac{\vert v_*-v \vert^2}{4\kappa^2}\right) \left(\frac{(v_*-v)}{\vert v_*-v \vert}\cdot\omega\right)^2 \nonumber\\
&\hspace{75mm} + 2 \left(1+\frac{\vert v_*-v \vert^2}{\kappa^2}\right) \det\left( \frac{(v_*-v)}{\vert v_*-v \vert},\omega \right)^2 \Bigg],
\end{align}
and writing:
\begin{align}
\frac{(v_*-v)}{\vert v_*-v \vert} \cdot \omega = \cos\theta,
\end{align}
we have
\begin{align}
\det\left( \frac{(v_*-v)}{\vert v_*-v \vert},\omega \right) = \sin\theta,
\end{align}
so that
\begin{align*}
\det(2A) &= \frac{4 \kappa^2}{\vert v_*-v \vert^2} \Bigg[ -1 - \left( 1+\frac{\vert v_*-v \vert}{4\kappa^2} \right) + 2 \left(1-\frac{\vert v_*-v \vert^2}{4\kappa^2}\right) \cos^2 \theta + 2 \left(1+\frac{\vert v_*-v \vert^2}{\kappa^2}\right) \sin^2\theta \Bigg] \nonumber\\
&= \frac{4 \kappa^2}{\vert v_*-v \vert^2} \Bigg[ -1 + \left( 1+\frac{\vert v_*-v \vert}{4\kappa^2} \right) \Bigg] = 1.
\end{align*}
\end{proof}
\noindent
With Proposition \ref{PROPOPrdSctcVV'TCTProceIn} and Theorem \ref{THEORSect2ConservatiMesur}, through \eqref{EQUATTheo1EvoluMesur_TCT_} we have now a complete understanding of the evolution of the Lebesgue measure under the action of the TCT dynamics, in the case of the IHSE model.

\begin{remar}
Let us observe that the flow of the IHSE system is a composition of free transport, and scatterings (given by \eqref{EQUATSect1LoideCollision_}, or \eqref{EQUATSect1LoideColliElast}), and such scattering mappings preserve the Lebesgue measure in the phase space. From such a point of view, this particle system is very similar to the classical, elastic hard sphere system. However, the measure-preserving property of the scattering is \emph{not enough} to conclude that the flow of such particle systems preserves the measure in the phase space.  Actually, Proposition \ref{PROPOPrdSctcVV'TCTProceIn} and Theorem \ref{THEORSect2ConservatiMesur} together prove that there exist particle systems such that their flows is not measure-preserving, even if their scattering mappings can be.\\
Formula \eqref{EQUATTheo1EvoluMesur_TCT_} of Theorem \ref{THEOREvoluMesur_TCT_Proce} shows that not only the scattering is playing a role in the evolution of the Lebesgue measure under the action of the flow of a TCT dynamics. Another important quantity is $\nabla_{X_N}t_c \cdot (V_N-V'_N)$, exhibiting that the link between the scattering and the collision time plays an important role in the change of the Lebesgue measure along the flow.
\end{remar}

\section{Consequences and interpretation of Theorem \ref{THEORSect2ConservatiMesur}}
\label{SECTI__5__}

\subsection{Global well-posedness of the flow of the inelastic hard spheres with emission}

In the case of the classical model of elastic hard spheres (see for instance \cite{Szas000}), the question of the global well-posedness of the dynamics of the particles is addressed by Alexander's theorem \cite{Alex975}, \cite{Alex976} (see also \cite{GSRT013} for a modern presentation). Such a well-posedness property for almost every initial configuration is the first step in order to complete the proof of Lanford's theorem \cite{Lanf975}, which provides a rigorous derivation of the Boltzmann equation taking as a starting point the elastic hard sphere system. This well-posedness is a delicate question, because particles can experience triple collisions, which causes problems to define further the dynamics, therefore such a dynamics has no chance to be globally well-posed for every initial configuration. Nevertheless, Alexander's theorem establishes such a result, for \emph{almost every} initial datum (with respect to the Lebesgue measure in the phase space of $N$ particles).\\
Since one of the main ingredients in Alexander's proof is the estimation of the Lebesgue measure of the preimages by the TCT dynamics, as a direct consequence of Theorem \ref{THEORSect2ConservatiMesur}, we can adapt such a proof in the case of the IHSE.\\
Let us recall the main steps of Alexander's proof. In order to define the dynamics of the particle system on the time interval $]0,T]$, one starts with decomposing this time interval into $T/\delta$ smaller intervals of the form $]k\delta,(k+1)\delta]$. Then, one identifies a set of initial configurations for which the dynamics is well-defined, and such that its complement is of small measure. We emphasize that, since the dynamics is not defined yet, one \emph{cannot} consider the preimages of the configurations such that three particles or more are in contact. The main idea is to observe that, except if two different pairs of particles are at a distance of order $\delta$, then at most one collision can occur on the time interval $]0,\delta]$. We denote by $E_0$ the set of such pathological initial configurations. Using the terminology we used in the present article, the transport can be defined either as the free flow, or as a TCT dynamics until time $\delta$ on the complement of $E_0$.  On the other hand, the pathological set $E_0$ has a measure smaller than $\delta^2$. 
Since now the flow is defined up to time $\delta$ for all the configurations that are not in $E_0$ (let us denote its flow by $\mathcal{T}_\delta$), we can repeat the procedure at time $\delta$, obtaining a pathological set $E_1$  of measure $\delta^2$ that could lead to triple collisions on the time interval $]\delta,2\delta]$. Therefore, outside $E_0$ and $\mathcal{T}_\delta^{-1}(E_1)$ (where $\mathcal{T}_\delta^{-1}(E_1)$ is the preimage of the set $E_1$ by the flow $\mathcal{T}_\delta$), the dynamics is now defined on $]0,2\delta]$. Besides, $\mathcal{T}_\delta^{-1}(E_1)$ and $E_1$ have the same measure, because the transport preserves the measure in the case of the elastic hard spheres. Proceeding until the final time $T$, we would define the dynamics on the whole interval $]0,T]$, up to exclude $T/\delta$ pathological sets of initial configurations, each of measure $\delta^2$, so that the dynamics is defined on $]0,T]$ up to exclude a set of measure $\delta$. Considering in the end the intersection for all $\delta>0$, we would define the elastic hard sphere dynamics almost everywhere in the phase space.\\
\newline
For our model, there are two main difficulties. First, in the present case, we defined a flow of particles that is not injective in the phase space. Therefore, each pathological set at time $(k+1)\delta$, has a preimage at time $k\delta$ that can be the union of two different sets (one as the preimage of a TCT dynamics with an elastic collision, the other of another TCT dynamics, with an inelastic collision). Second, the measure is not conserved by the TCT dynamics in the case of an inelastic collision. In particular, \eqref{EQUATPrdSctcVV'TCTProceIn} implies the following. If $A$ is a measurable subset of the phase space such that $\vert v_i - v_j \vert^2 - 4\varepsilon_0$ is small for all $Z_N=(x_1,\dots,x_N,\dots,v_i,\dots,v_j,\dots) \in A$, the Lebesgue measure of $\mathcal{T}_\delta(A)$ is much smaller than the measure of $A$ itself.\\
The first difficulty can be addressed using the following key argument: there exists an a priori uniform estimate on the number of collisions in a system of inelastic hard spheres with emission, as soon as the kinetic energy is bounded. This estimate is a consequence of the uniform bounds that are known for the systems of elastic hard spheres, obtained by Burago, Ferleger and Kononenko (Theorem 1.3 in \cite{BuFK998}). We recall here their result:

\begin{theor}[Burago, Ferleger, Kononenko \cite{BuFK998}]
\label{THEORBuragFerleKonon}
Let $d \geq 2$ and $N$ be two positive integers.\\
To any initial configuration $Z_N \in\overline{\mathcal{D}}_N$, we can associate a non-negative real number $T(Z_N) \in [0,+\infty]$ such that the maximal interval of existence of the trajectory starting from $Z_N$ is $[0,T(Z_N)[$.\\
In addition, there exists a positive constant $C_\text{HS} = C_\text{HS}(N)$, depending only $N$, such that for any initial configuration $Z_N$ of $N$ elastic hard spheres, the trajectory starting from $Z_N$ presents at most $C_\text{HS}$ collisions on the whole time interval $[0,T(Z_N)[$.
\end{theor}

\begin{remar}
The theorem of Burago, Ferleger and Kononenko is stated for any initial configuration $Z_N \in \overline{\mathcal{D}}_N$. In particular, there exist some initial configurations $Z_N \in \overline{\mathcal{D}}_N$ generating trajectories that are not globally defined, because triple or higher order collisions might occur. 
In such a case the dynamics introduced in \cite{BuFK998} is not defined further after a collision that is not binary. Nevertheless, the estimate on the number of collisions holds on the whole maximal time interval of definition $[0,T(Z_N)]$ of the trajectories.\\
Notice that Alexander's theorem for elastic hard spheres implies that $T(Z_N) = +\infty$ for almost every initial configuration $Z_N$ in $\overline{\mathcal{D}}_N$.
\end{remar}

\noindent
The second difficulty will be addressed by introducing a cut-off in the phase space. More precisely, if we assume, for $0 < \mu < 1$ (where $\mu$ is meant to be a small cut-off parameter), that:
\begin{align}
\label{EQUATCutOfIntro_Mu__}
\frac{4\varepsilon_0}{1 - \mu} \leq \vert v_i - v_j \vert^2,
\end{align}
then we have $\displaystyle{\sqrt{1 - \frac{4 \varepsilon_0}{\vert v_i-v_j \vert^2}} > \sqrt{\mu}}$, so that:
\begin{align}
\label{EQUATDescrJacobTrans}
\left\vert \det\left( \text{Jac}(\mathcal{T}_\tau) \right) \right\vert \geq \sqrt{\mu}.
\end{align}
Alexander's proof starts with establishing global well-posedness of the trajectories for almost every initial datum $Z_N$ such that
\begin{align}
\label{EQUATIntroCutOfR1_R2}
\vert X_N \vert \leq R_1 \hspace{3mm} \text{and} \hspace{3mm} \vert V_N \vert \leq R_2,
\end{align}
with $0 < R_1,R_2 < +\infty$. At the end of the argument, $R_1$ and $R_2$ are sent to infinity. 
We will proceed in the same way. More precisely, the parameters $R_1$ and $R_2$, that will be used extensively in the intermediate results that follow, will play the role of cut-offs for the variables $X_N$ and $V_N$ respectively.\\
The cut-off parameter $\delta$ will be used, as in the proof of Alexander's theorem, as the time discretization parameter. Besides, for $\delta > 0$ fixed, we will introduce another cut-off parameter $\mu$ and define a region $\mathcal{P} = \mathcal{P}(\delta,\mu)$ in the phase space by means of the condition \eqref{EQUATCutOfIntro_Mu__}. This condition states that the difference between the norm of the relative velocity of a pair of colliding particles and $2\sqrt{\varepsilon_0}$ is large enough. In the opposite case, that is, when \eqref{EQUATCutOfIntro_Mu__} does not hold, for such a colliding pair, the relative post-collisional velocity is small. $\mathcal{P}$ will be defined as the subset of the phase space such that \eqref{EQUATCutOfIntro_Mu__} does not hold.\\
We will show that the measure of the region $\mathcal{P}$ will be small, and any trajectory starting from the complement of $\mathcal{P}$ satisfies the cut-off \eqref{EQUATCutOfIntro_Mu__}. As a consequence, since the cut-off \eqref{EQUATCutOfIntro_Mu__} implies \eqref{EQUATDescrJacobTrans}, the preimage of a set of small measure by the flow, intersected with the complement of $\mathcal{P}$, has also a small measure.\\
We introduce also, as in the proof of Alexander's theorem, another pathological set $\mathcal{A}$, that contains the initial configurations $Z_N$ such that two collisions or more can take place in any time interval of the form $]k\delta,(k+1)\delta]$. In this way, we prove global well-posedness of the flow outside the region $\mathcal{P} \cup \mathcal{A}$. $\mathcal{P}$ and $\mathcal{A}$ having respective measures that vanish when $\delta,\mu \rightarrow 0$, we will define the flow almost everywhere, provided that $X_N$ and $V_N$ are estimated as in \eqref{EQUATIntroCutOfR1_R2}. We will conclude the argument by sending the cut-offs $R_1$ and $R_2$ to infinity to recover a domain in the phase space, of full Lebesgue measure, on which the flow is well-defined for all positive times.\\
In summary, result we obtain is the following.

\begin{theor}[Alexander's theorem for the inelastic hard sphere with emission model]
\label{THEORSect3AlexanderIHSWEM}
Let $N$ be any positive integer, and let $\varepsilon_0 > 0$ be a positive real number.\\ Then, the dynamics of the system of $N$ inelastic hard spheres with emission, introduced in Section \ref{SSec1PresentatiModèl}, is almost everywhere (with respect to the Lebesgue measure in the phase space $\mathcal{D}_N$) globally well-defined in dimension $d = 2$.\\ 
Moreover, for almost every initial configuration in the phase space, this dynamics is defined only in terms of binary collisions.\\
In addition for any $T > 0$, and for almost every initial configuration $Z_N \in \mathcal{D}_N$, the system of particles starting initially from $Z_N$ experiences only a finite number of collisions in the time interval $[0,T]$.
\end{theor}

\noindent
We divide the proof of Theorem \ref{THEORSect3AlexanderIHSWEM} in the following steps. For any positive number $\delta > 0$, we start with decomposing the time interval $]0,T]$ into the intervals $]k\delta,(k+1)\delta]$, of length $\delta$.
\begin{itemize}
\item First, we introduce the different subsets of the phase space that allow to construct the dynamics of the particles recursively, on the time intervals $]k\delta,(k+1)\delta]$, one after the other, in order to cover the whole time interval $[0,T]$.
\item In a second time, we will identify a partition of the phase space into cells, such that the flow is injective on each of these cells, and such that the number of cells of this partition is independent from $\delta$.
\item Third, we will estimate the measure of the pathological sets we introduced in the first step. The measure of the pathological sets being given explicitely in terms of an expression that vanishes when $\delta \rightarrow 0$, we will be able to conclude the proof of the theorem.
\end{itemize}

\noindent
We start with introducing the two families of pathological sets, on the complement of which we will define the IHSE dynamics.

\begin{defin}[Pathological sets $E_k(\delta)$, $P_k(\delta,\mu)$]
\label{DEFINPatho_Sets_E_P_}
Let $N$ be a positive integer and let $\varepsilon_0 > 0$ be a positive real number. Let $R_1,R_2 > 0$ be two positive real numbers. Let $k \geq 0$ be a non-negative integer, and $\delta,\mu > 0$ be two positive real numbers.\\
We define the \emph{pathological sets of possible multiple collisions}, denoted by $E_k(\delta)$, as:
\begin{align}
\label{EQUATDefinEnsmbPatho_E_kd}
E_k(\delta) = \mathcal{D}_N \cap \big\{ Z_N \in B_{X_N}(0,R_1+k\delta R_2)\times &B_{V_N}(0,R_2)\ /\ \exists\,(i,j),(l,m) \in \mathcal{P}_N\ /\ (i,j)\neq (l,m) \text{ and }\nonumber\\
&\hspace{10mm} \vert x_i-x_j \vert \leq 1 + \frac{3}{2}\sqrt{2}\delta R_2, \vert x_l-x_m \vert \leq 1 + \frac{3}{2}\sqrt{2}\delta R_2 \big\}.
\end{align}
We define also the \emph{pathological sets of large Jacobian}, denoted by $P_k(\delta,\mu)$, as:
\begin{align}
\label{EQUATDefinEnsmbPatho_P_kd}
P_k(\delta,\mu) &= \mathcal{D}_N \cap \big\{ Z_N \in B_{X_N}(0,R_1+k\delta R_2) \times B_{V_N}(0,R_2)\ /\ \exists\, (i,j) \in \mathcal{P}_N \text{ and }\nonumber\\
&\hspace{32mm}\vert x_i-x_j \vert \leq 1 + \sqrt{2} \delta R_2,\hspace{1mm} 2\sqrt{\varepsilon_0} \leq \vert v_i-v_j \vert \leq 2\sqrt{\varepsilon_0}\left( 1 +(\sqrt{2}-1)\mu \right) \big\}.
\end{align}
Finally, we introduce the set of \emph{grazing initial configurations}, denoted by $\Delta$, as:
\begin{align}
\label{EQUATDefinEnsmbPathoDelta}
\Delta = \{Z_N \in \mathcal{D}_N\ /\ \exists\, (i,j) \in \{1,\dots,N\}^2,\, i<j \ /\ \Delta_{i,j} = 0 \text{ or } \vert v_i-v_j \vert = 2\sqrt{\varepsilon_0} \},
\end{align}
and $\Delta_{i,j}$ defined in \eqref{EQUATHypotColli_Non_Rasan}.
\end{defin}

\noindent
The set $E_k(\delta)$ is the set of the initial configurations $Z_N$ such that at least two pairs of particles are initially at a distance smaller than $\frac{3}{2}\sqrt{2}\delta R_2$ from each other.\\
The set $P_k(\delta,\mu)$ is the set of initial configurations such that there exists at least one pair of particles initially at a distance smaller than $\sqrt{2}\delta R_2$, and such that the norm of the relative velocity of such a pair is close to $2\sqrt{\varepsilon_0}$, so that if a collision takes place involving this pair, the post-collisional relative velocity will be small.\\
\newline
\noindent
We begin proving that the dynamics  of the IHSE is well-defined on the time interval $0 \leq t \leq \delta$, for any initial configuration $Z_N$ that does not belong to $E_0(\delta)$, where $E_0(\delta)$ is defined according to Definition \ref{DEFINPatho_Sets_E_P_}, taking $k = 0$.

\begin{lemma}[Well-posedness of the IHSE dynamics on the complement of $E_0(\delta) \cup \Delta$]
\label{LEMMEBonneDefin_IHSE_0_d_}
Let $N$ be a positive integer and let $\varepsilon_0 > 0$ be a positive real number. Let $R_1,R_2 > 0$ be two positive real numbers. Let $\delta > 0$ be a positive real number.\\
Then, we have:
\begin{align}
\mathcal{D}_N \cap \left[ \left( B_{X_N}(0,R_1) \times B_{V_N}(0,R_2) \right) \backslash \left( E_0(\delta) \cup \Delta \right) \right] \subset \left( \mathcal{D}_N^{(0)}(\delta) \cup \mathcal{D}_N^{(1)}(\delta) \right),
\end{align}
where $E_0(\delta)$ and and $\Delta$ are defined by \eqref{EQUATDefinEnsmbPatho_E_kd} and \eqref{EQUATDefinEnsmbPathoDelta} respectively.\\
Therefore, the IHSE dynamics is well-defined on  the time interval $[0,\delta]$ for any initial configuration that belongs to the subset $\mathcal{D}_N \cap \left[ \left( B_{X_N}(0,R_1) \times B_{V_N}(0,R_2) \right) \backslash \left( E_0(\delta) \cup \Delta \right) \right]$ of the phase space $\mathcal{D}_N$. We will denote by $\mathcal{T}_\delta$ such a dynamics:
\begin{align}
\mathcal{T}_\delta : \mathcal{D}_N \cap \left[ \left( B_{X_N}(0,R_1) \times B_{V_N}(0,R_2) \right) \backslash \left( E_0(\delta) \cup \Delta \right) \right] \rightarrow \mathcal{D}_N.
\end{align}
More precisely, if $Z_N \in \mathcal{D}_N^{(0)}$, $\mathcal{T}_\delta(Z_N)$ is defined according to Definition \ref{DEFINDomai&TransportLibre}, and if $Z_N \in \mathcal{D}_N^{(1)}$, $\mathcal{T}_\delta(Z_N)$ is defined according to Definition \ref{DEFIN_TCT_ProceInelaHSwEm} (taking $\tau = \delta$ in Definitions \ref{DEFINDomai&TransportLibre} and \ref{DEFIN_TCT_ProceInelaHSwEm}). 
\end{lemma}

\begin{proof}
let $Z_N$ be an initial configuration outside $E_0(\delta)$. 
Then, such initial configuration $Z_N$ leads to a well-defined trajectory on the time interval $[0,\delta]$, because such an initial configuration generates a trajectory with at most one collision on the time interval $]0,3\delta/2[$, and so, we may define the dynamics on $]0,\delta]$ with the help, either of the free flow (Definition \ref{DEFINDomai&TransportLibre}), or the TCT dynamics (Definitions \ref{DEFIN_TCT_ElemnFonctInelaHSwEm} and \ref{DEFIN_TCT_ProceInelaHSwEm}).
\end{proof}
\noindent
\noindent
As a next step, we evaluate the Lebesgue measure of the pathological sets $E_k(\delta)$ and $P_k(\delta,\mu)$.

\begin{lemma}[Estimation of the measures of the pathological sets $E_k(\delta)$ and  $P_k(\delta,\mu)$]
\label{LEMMEMesurEnsmbPatho_E_kd_P_kd}
Let $N$ be a positive integer and let $\varepsilon_0 > 0$ be a positive real number. Let $T,R_1,R_2 > 0$ be three positive real numbers. Let $\delta,\mu > 0$ be two positive real numbers, such that $T/\delta$ is a positive integer $n > 0$.\\
Let us assume that:
\begin{align}
\delta \leq 1,\hspace{5mm} \delta \leq \frac{2}{3\sqrt{2}R_2},\hspace{5mm} \mu \leq \frac{1}{2}\cdotp
\end{align}
Then, in dimension $d=2$, there exist two positive constants $C_1=C_1(N,T,R_1,R_2) > 0$ and $C_2=C_2(N,T,\varepsilon_0,R_1,R_2) > 0$ (that depend on the parameters $N$, $\varepsilon_0$, $T$, $R_1$ and $R_2$, but that are independent from $\delta$ and $\mu$) such that, for any integer $0 \leq k \leq n-1$, the Lebesgue measures of the pathological sets $\vert E_k \vert$ and $\vert P_k \vert$ is estimated as follows:
\begin{align}
\label{EQUATSect3MesurEnsmbEdelt}
\left\vert E_k(\delta) \right\vert &\leq C_1(N,T,R_1,R_2)\delta^2,
\end{align}
and:
\begin{align}
\label{EQUATSect3MesurEnsmbPdelt}
\left\vert P_k(\delta,\mu) \right\vert \leq C_2(N,T,\varepsilon_0,R_1,R_2) \delta \mu.
\end{align}
\end{lemma}

\begin{proof}
Let $0 \leq k \leq n-1$ be a non-negative integer, where $n = T/\delta$. Concerning the pathological set of multiple collisions $E_k(\delta)$, we have:
\begin{align}
\left\vert E_k(\delta) \right\vert &\leq {N \choose 2}^2 \left\vert B_{\mathbb{R}^{2(N-2)}}(0,R_1+k\delta R_2) \right\vert \times \Big( \Big\vert B_{\mathbb{R}^2}(0,1 + \frac{3\sqrt{2}}{2}\delta R_2) \Big\vert - \left\vert B_{\mathbb{R}^2}(0,1) \right\vert \Big)^2 \times \left\vert B_{\mathbb{R}^{2N}}(0,R_2) \right\vert \nonumber \\
&\leq C(N) (R_1 + k\delta R_2)^{2(N-2)} R_2^{2N+2} \delta^2 \nonumber \\
&\leq C(N) (R_1 + T R_2)^{2(N-2)} R_2^{2N+2} \delta^2 \nonumber
\end{align}
In particular, we have:
\begin{align}
\vert E_k(\delta) \vert \leq C_1(N,T,R_1,R_2) \delta^2.
\end{align}
The estimate on $\vert P_k(\delta,\mu) \vert$ is obtained in the same way:
\begin{align}
\left\vert P_k(\delta,\mu) \right\vert = {N \choose 2}& \left\vert B_{\mathbb{R}^{2(N-1)}}(0,R_1+k\delta R_2)\right\vert \times \left( \Big\vert B_{\mathbb{R}^2}(0,1 + \frac{3\sqrt{2}}{2}\delta R_2) \Big\vert - \left\vert B_{\mathbb{R}^2}(0,1) \right\vert \right) \nonumber\\
&\times \left\vert B_{\mathbb{R}^{2(N-1)}}(0,R_2)\right\vert \times \left( \Big\vert B_{\mathbb{R}^2}(0,2\sqrt{\varepsilon_0} + 2(\sqrt{2}-1)\sqrt{\varepsilon_0}\mu) \Big\vert - \left\vert B_{\mathbb{R}^2}(0,2\sqrt{\varepsilon_0}) \right\vert \right),
\end{align}
so that we find:
\begin{align}
\left\vert P_k(\delta,\mu) \right\vert \leq C(N) (R_1+TR_2)^{2(N-1)} R_2^{2(N-1)+1} \delta \sqrt{\varepsilon_0}\mu,
\end{align}
which concludes the proof of Lemma \ref{LEMMEMesurEnsmbPatho_E_kd_P_kd}.
\end{proof}

\noindent
We now introduce two additional pathological subsets $\mathcal{A}(\delta)$ and $\mathcal{P}(\delta,\mu)$ of the phase space $\mathcal{D}_N$. Outside the first pathological subset $\mathcal{A}(\delta)$, we will be able to define the IHSE dynamics on any time interval $[0,T]$. Outside the second pathological subset $\mathcal{P}(\delta,\mu)$, we will estimate the Jacobian of the IHSE dynamics, which will be necessary in turn to estimate the measure of the pathological set $\mathcal{A}(\delta)$. We first define recursively $\mathcal{A}(\delta)$, together with the IHSE dynamics on the complement of $\mathcal{A}(\delta)$, on the time interval $[0,T]$.\\
According to Lemma \ref{LEMMEBonneDefin_IHSE_0_d_}, the IHSE dynamics is well-defined on the complement of $E_0(\delta) \cap \Delta$ in the product $B_{X_N}(0,R_1) \times B_{V_N}(0,R_2)$. Due to the dispersion of the particles, we know that the image of the product $B_{X_N}(0,R_1) \times B_{V_N}(0,R_2)$ by the IHSE dynamics $\mathcal{T}_\delta$ lies in $B_{X_N}(0,R_1+\delta R_2) \times B_{V_N}(0,R_2)$. If in addition $\mathcal{T}_\delta(Z_N)$ does not belong to $E_1 \cup \Delta$, then it is possible to define the image of $\mathcal{T}_\delta(Z_N)$ by the IHSE dynamics. In other words, this dynamics is now defined on $[0,2\delta]$, provided that $Z_N \in \left( B_{X_N}(0,R_1) \times B_{V_N}(0,R_2) \right) \cap \mathcal{D}_N$ and:
\begin{align}
Z_N \notin E_0 \cup \Delta \cup \mathcal{T}_\delta^{-1}\left( E_1 \cup \Delta \right).
\end{align}
This observation motivates the following definition.

\begin{defin}[IHSE dynamics, final pathological subset of possible multiple collisions $\mathcal{A}(\delta)$]
\label{DEFINBonneDefin_IHSE_0_kd_A(d)}
Let $N$ be a positive integer and let $\varepsilon_0 > 0$ be a positive real number. Let $T,R_1,R_2 > 0$ be three positive real numbers. Let $\delta,\mu > 0$ be two positive real numbers, such that $T/\delta$ is a positive integer $n > 0$.\\
By convention we define $\mathcal{T}_0$ as the identity mapping on the phase space $\mathcal{D}_N$.\\
For any $0 \leq k \leq n-1$, we define recursively the set $F_k$ and the IHSE dynamics $\mathcal{T}_{k\delta}$ on the time interval $[0,k\delta]$ as follows:
\begin{itemize}
\item $F_k \subset \left[ B_{X_N}(0,R_1) \times B_{V_N}(0,R_2) \right] \cap \mathcal{D}_N$ is the subset of the phase space defined as
\begin{align}
F_k = \Big( \bigcap_{j=0}^{k-1} F_j^c \Big) \cup \mathcal{T}_{k\delta}^{-1}( E_k \cup \Delta ),
\end{align}
\item $\mathcal{T}_{(k+1)\delta}$ is the mapping defined from $\bigcap_{j=0}^k F_j^c \subset \left[ B_{X_N}(0,R_1) \times B_{V_N}(0,R_2) \right] \cap \mathcal{D}_N$, taking values in $\left[ B_{X_N}(0,R_1+k\delta R_2) \times B_{V_N}(0,R_2) \right] \cap \mathcal{D}_N$, as:
\begin{align}
\mathcal{T}_{(k+1)\delta}(Z_N) = \mathcal{T}_\delta\left(\mathcal{T}_{k\delta}(Z_N)\right) \hspace{3mm} \forall Z_N \in \bigcap_{j=0}^k F_j^c,
\end{align}
with $\mathcal{T}_\delta$ being the IHSE dynamics on the time interval $[0,\delta]$ introduced in Lemma \ref{LEMMEBonneDefin_IHSE_0_d_}.
\end{itemize}

\noindent
We introduce then the \emph{final pathological set of multiple collisions}, denoted by $\mathcal{A} = \mathcal{A}(\delta)$, as:
\begin{align}
\label{EQUATDefinEnsmbPatho_A(d)}
\mathcal{A}(\delta) = \bigcup_{k = 0}^{n-1} F_k \subset \left( B_{X_N}(0,R_1) \times B_{V_N}(0,R_2) \right) \cap \mathcal{D}_N.
\end{align}
\end{defin}

\noindent
In order to introduce the final pathological set $\mathcal{P}(\delta,\mu_1,\dots,\mu_n)$ of large Jacobian, we will partition the domain of the IHSE dynamics into cells, that contain the initial configurations leading to trajectories experiencing inelastic collisions during the same time intervals $]j\delta,(j+1)\delta[$.

\begin{defin}[First partition of the domain of the IHSE dynamics]
Let $N$ be a positive integer and let $\varepsilon_0 > 0$ be a positive real number. Let $T,R_1,R_2 > 0$ be three positive real numbers. Let $\delta,\mu > 0$ be two positive real numbers, such that $T/\delta$ is a positive integer $n > 0$.\\
For any integer $0 \leq k \leq n-1$, and for any subset $R \subset \{1,\dots,k+1\}$, we define the \emph{cell of first kind, of type $(k,R)$}, denoted by $C_k^R$, as the subset of $\left(\bigcup_{j=0}^k F_j\right)^c$ defined as:
\begin{align}
C_k^R = \big\{ &Z_N \in \Big(\bigcup_{j=0}^k F_j\Big)^c \ /\ \forall\, l \in \{1,\dots,k+1\},\text{ the trajectory originating from }Z_N\text{ presents an inelastic}\nonumber\\
&\hspace{40mm}\text{collision during the time interval } ](l-1)\delta,l\delta[ \text{ if and only if } l \in R\big\}.
\end{align}
\end{defin}

\noindent
The subsets $C_k^R$ form a partition of $\left(F_0 \cup \dots \cup F_k\right)^c$. Observe that $R$ can be the empty set.\\
We are now in position to introduce the final pathological set $\mathcal{P}(\delta,\mu_1,\dots,\mu_n)$, where the Jacobian of the TCT dynamics might be large.

\begin{defin}[Final pathological set $\mathcal{P}(\delta,\mu_1,\dots,\mu_n)$ of large Jacobian]
\label{DEFINEnsmbPatho_P(d)}
Let $N$ be a positive integer and let $\varepsilon_0 > 0$ be a positive real number. Let $T,R_1,R_2 > 0$ be three positive real numbers. Let $\delta,\mu > 0$ be two positive real numbers, such that $T/\delta$ is a positive integer $n > 0$, and let finally $(\mu_1,\dots,\mu_n)$ be a finite family of $n$ positive real numbers.\\
We define the \emph{final pathological set of large Jacobian}, denoted by $\mathcal{P}(\delta,\mu_1,\dots,\mu_n)$, as:
\begin{align}
\mathcal{P}(\delta,\mu_1,\dots,\mu_n) = \bigcup_{k=0}^{n-1} G_k,
\end{align}
with
\begin{align}
G_0^\emptyset = \emptyset,\hspace{5mm} G_0^{\{1\}} = P_0(\delta,\mu_1),
\end{align}
and recursively, for any $1 \leq k \leq n-1$:
\begin{align}
\label{EQUATDefinEnsmb_G_k_delta}
G_k = \hspace{-6mm}\bigcup_{\substack{ R \subset \{1,\dots,k+1\} \\ (k+1) \in R}} \hspace{-6mm}G_k^R.
\end{align}
with
\begin{align}
G_k^R = \mathcal{T}_{k\delta}^{-1}\left(P_k(\delta,\mu_{\vert R \vert})\right) \cap \left( \left(G_0^{R \cap \{1\}}\right)^c \cap \left(G_1^{R \cap \{1,2\}}\right)^c \cap \dots \cap \left( G_{(k-1)}^{R \cap \{1,\dots,k\}}\right)^c \right) \cap C_k^R.
\end{align}
\end{defin}

\noindent
In order to illustrate the complicated definition of the sets $G_k$, we will present completely the decomposition of $G_1$, which is the first non trivial set of this family. We define:
\begin{align}
G_1 = \underbrace{\left(\mathcal{T}_\delta^{-1}(P_1(\delta,\mu_1)) \cap C_\delta^\emptyset\right)}_{G_1^{\{2\}}} \cup \underbrace{\left(\mathcal{T}_\delta^{-1}(P_1(\delta,\mu_2)) \cap P_0(\delta,\mu_1)^c \cap C_\delta^{\{1\}}\right)}_{G_1^{\{1,2\}}}.
\end{align}
where $\mu_2$ is the second element of the family $(\mu_i)_i$. $G_1$ is the union of the two following subsets of the phase space:
\begin{itemize}
\item $G_1^{\{2\}}$ contains the initial configurations for which a single inelastic collision takes place, during the time interval $]\delta,2\delta]$, and such that the Jacobian determinant is not estimated on this interval. \item $G_1^{\{1,2\}}$ contains the initial configurations such that two inelastic collisions take place, in $]0,\delta]$, and in $]\delta,2\delta]$, and such that the Jacobian determinant is estimated in the interval $]0,\delta]$ (due to the term $P_0(\delta,\mu_1)^c$~; following the notations we introduced, note that such a term can then be written as $P_0(\delta,\mu_1) = G_0^{\{1\}}$), but the Jacobian is not estimated in the interval $]\delta,2\delta]$.
\end{itemize}
Therefore, on $\left[ \left(F_0 \cup F_1 \right) \cup \left( G_0 \cup G_1 \right) \right]^c$, the Jacobian determinant of $\mathcal{T}_{2\delta}$ is larger than $\sqrt{\mu_1}$ in $C_\delta^\emptyset$, while it is larger than $\sqrt{\mu_1\mu_2}$ in $C_\delta^{\{1\}}$.\\
\newline
More generally, we can state the following result.

\begin{lemma}[Estimate of the Jacobian of the TCT dynamics in the cells $C_k^R$]
\label{LEMMEEstimJacob}
Let $N$ be a positive integer and let $\varepsilon_0 > 0$ be a positive real number. Let $T,R_1,R_2 > 0$ be three positive real numbers. Let $\delta,\mu > 0$ be two positive real numbers, such that $T/\delta$ is a positive integer $n > 0$, and let finally $(\mu_1,\dots,\mu_n)$ be a finite family of $n$ positive real numbers.\\
For any integer $k \in \{0,\dots,n-1\}$ and any subset $R$ of $\{1,\dots,k+1\}$ such that $(k+1) \in R$, we have:
\begin{align}
&\hspace{50mm}\det\left(\mathcal{T}_{k\delta}(Z_N)\right) \geq \sqrt{\mu_1 \dots \mu_{(\vert R \vert-1)}} \\
&\hspace{60mm}\forall\, Z_N \in \left( \left(G_0^{R \cap \{1\}}\right)^c \cap \left(G_1^{R \cap \{1,2\}}\right)^c \cap \dots \cap \left( G_{(k-1)}^{R \cap \{1,\dots,k\}}\right)^c \right) \cap C_k^R, \nonumber
\end{align}
where $\vert R \vert$ denotes the cardinal of the set $R$.
\end{lemma}

\noindent
To complete the proof of Theorem \ref{THEORSect3AlexanderIHSWEM}, we will need two additional tools, namely, another partition of the domain of the IHSE dynamics, and a general result of set theory.

\begin{defin}[Second partition of the domain of the IHSE dynamics]
\label{DEFINDeuxiPartiCelluC^Qpq}
Let $N$ be a positive integer and let $\varepsilon_0 > 0$ be a positive real number. Let $T,R_1,R_2 > 0$ be three positive real numbers. Let $\delta,\mu > 0$ be two positive real numbers, such that $T/\delta$ is a positive integer $n > 0$.\\
For any integers $0 \leq p \leq n$ and $0 \leq q \leq p$, and any subset $Q \subset \{1,\dots,p\}$ with $\vert Q \vert = q$, we define the \emph{cell of second kind, of type $(p,q,Q)$}, denoted by $L_{p,q}^Q$, as the subset of $\left(\bigcup_{j=0}^{n-1} F_j\right)^c$ defined as:
\begin{align}
L_{p,q}^Q = \Big\{ Z_N \in F_0^c \cap \dots \cap F_{(k-1)}^c\ /\ &t \mapsto \mathcal{T}_t(Z_N) \text{ has $p$ collisions on $]0,k\delta[$, $q$ of them are inelastic,} \nonumber\\
&\hspace{5mm} \text{and the label of such inelastic collisions are the elements of $Q$.}\Big\}
\end{align}
\end{defin}

\begin{remar}
Knowing $Q$ does not allow to know directly in which intervals $]j\delta,(j+1)\delta]$ take place the inelastic collisions. In this regard, the cells $L_{p,q}^Q$ are very different from the cells $C_k^R$. Restricted on such cells $L_{p,q}^Q$, the transport $Z_N \mapsto \mathcal{T}_T(Z_N)$ is injective.\\
In addition, and it is one of the key arguments of the proof of Theorem \ref{THEORSect3AlexanderIHSWEM}, we will see that the number of possible cells $L_{p,q}^Q$ is smaller than a constant that does not depend on $\delta$.
\end{remar}

\begin{lemma}
Let $X$ and $Y$ be two sets, let $f:X \rightarrow Y$ be a function, and let $C \subset X$ and $P \subset Y$, such that $f$ is injective on $C$. Then:
\begin{itemize}
\item for any $A \subset X$ we have:
\begin{align}
f^{-1}(P) \cap A \cap C = f^{-1}\left( P \cap f(A \cap C) \right),
\end{align}
\item for any $(A_i)_{i \in I} \subset X^I$ with $I$ any set, and such that the $A_i$ are pairwise disjoint, we have:
\begin{align}
\label{EQUATTheorEnsmb2Decp__P__}
\bigcup_{i\in I} \left[ P \cap f(A_i \cap C) \right] = P \cap f( \Big( \bigcup_{i\in I} A_i \Big) \cap C ),
\end{align}
and the $P \cap f(A_i\cap C)$ are pairwise disjoint.
\end{itemize}
\end{lemma}

\noindent
After the previous preliminary results, we conclude now the proof of Theorem \ref{THEORSect3AlexanderIHSWEM}.

\begin{proof}[Proof of Theorem \ref{THEORSect3AlexanderIHSWEM}]
We consider any arbitrary time $T > 0$, and two cut-off parameters $R_1, R_2 > 0$. Let us denote by $Z_N = \left(x_1,v_1,\dots,x_N,v_N\right)$ the initial configurations of the system of $N$ particles, and we will assume that the initial positions and velocities are such that:
\begin{align}
\left\vert X_N \right\vert = \left\vert \left(x_1,x_2,\dots,x_N\right) \right\vert \leq R_1, \hspace{3mm} \left\vert V_N \right\vert = \left\vert \left(v_1,v_2,\dots,v_N\right) \right\vert \leq R_2.
\end{align}
We introduce a third cut-off parameter $\delta$, meant to be small, such that $T/\delta = n$ is an integer, and such that:
\begin{align}
\delta \leq 1,\hspace{3mm} \delta \leq \frac{2}{3\sqrt{2} R_2} \cdotp
\end{align}
Finally, we introduce a finite family of $n$ positive real numbers $(\mu_i)_i$ all smaller than $1/2$, to be chosen later.\\
According to Lemma \ref{LEMMEBonneDefin_IHSE_0_d_} and Definition \ref{DEFINBonneDefin_IHSE_0_kd_A(d)}, the IHSE dynamics is well-defined on the time interval $[0,T]$, for any initial configuration in $\mathcal{A}(\delta)$ (defined in \eqref{EQUATDefinEnsmbPatho_A(d)}). We will now show that the Lebesgue measure of $\mathcal{A}(\delta)$ is vanishing as $\delta$ goes to zero.\\
\newline
To estimate $\left\vert \mathcal{A}(\delta) \right\vert$ and $\left\vert \mathcal{P}(\mu_1,\dots,\mu_n) \right\vert$, we need an a priori bound on the number of collisions of the trajectories that we have defined. For any initial configuration $Z_N \notin \mathcal{A}(\delta)$, well-defined on $[0,T]$, the maximal number of \emph{inelastic} collisions is strictly smaller than $R_2^2/\varepsilon_0$. Let $K = K(\varepsilon_0,R_2)$ be the largest integer strictly smaller than $R_2^2/\varepsilon_0$.\\
In between two consecutive inelastic collisions, by definition the system can experience only elastic collisions. Now, there exists a universal bound, depending only on the number of particles $N$, denoted by $C_\text{HS}(N)$, on the maximal number of collisions that a system of elastic hard spheres can undergo (see Theorem \ref{THEORBuragFerleKonon}). Therefore, the trajectory starting from $Z_N$ can experience at most:
\begin{align}
K + (K+1) C_{\text{HS}} = C_\text{IHS}
\end{align}
collisions (elastic and inelastic). This bound is also universal, in the sense that it depends only on $N$ and $R_2$, but not on the initial configuration $Z_N$ (taken in $B_{X_N}(0,R_1) \times B_{V_N}(0,R_2)$), nor the time $T$, nor the cut-off parameter $\delta$.\\
As a consequence, since we will use the parameters $\mu_i$ to estimate the modification of the Jacobian determinant of the transport due to the different inelastic collisions, only $K$ parameters $\mu_i$ will be needed.\\
\newline
We start with estimating the measure of the pathological set $\mathcal{P}(\delta,\mu_1,\dots,\mu_K)$, introduced in Definition \ref{DEFINEnsmbPatho_P(d)}. We will estimate separately the measure of the sets $G_k$ introduced in \eqref{EQUATDefinEnsmb_G_k_delta}, by decomposing each of them on the cells $L_{p,q}^Q$ (introduced in Definition \ref{DEFINDeuxiPartiCelluC^Qpq}).\\
By definition, the sets $G_k$ are the unions of the sets $G_k^R$, that are themselves by definition subsets of the cells $C_k^R$. Since by definition a cell $C_k^R$ does not intersect the sets $F_0,\dots,F_{(k-1)}$, the IHSE dynamics $\mathcal{T}_{k\delta}$ is well-defined on the cells $C_k^R$.\\
So, let $0 \leq p \leq \min(n,C_\text{IHS})$ and $0 \leq q \leq \min(p,K)$ be two integers, and let $Q$ be a subset of $\{1,\dots,p\}$, defining a cell $L_{p,q}^Q$. Let $0 \leq k \leq n-1$ be another integer. We will decompose the set $G_k \cap L^Q_{p,q}$ into the sets $G_k^R \cap L_{p,q}^Q$, for all the non empty subsets $R \subset \{1,\dots,k+1\}$ that contain $k+1$ (so that, in particular, all the subsets $R$ that we consider are not empty). We have by definition:
\begin{align}
G_k \cap L_{p,q}^Q &= \hspace{-6mm}\bigcup_{\substack{ R \subset \{1,\dots,k+1\} \\ (k+1) \in R}} \hspace{-6mm} \left( G_k^R \cap L_{p,q}^Q \right) \nonumber\\
&\hspace{-10mm}= \bigcup_{j=1}^{k+1} \hspace{-0mm}\bigcup_{\substack{ R \subset \{1,\dots,k+1\} \\ (k+1) \in R \\ \vert R \vert = j}} \hspace{-6mm} \left( \left[ \mathcal{T}_{k\delta}^{-1}\left(P_k(\mu_{\vert R \vert})\right) \cap \left( \left(G_0^{R \cap \{1\}}\right)^c \cap \left(G_1^{R \cap \{1,2\}}\right)^c \cap \dots \cap \left( G_{(k-1)}^{R \cap \{1,\dots,k\}}\right)^c \right) \cap C_k^R \right] \cap L_{p,q}^Q \right).
\end{align}
In order to simplify the notations, let us denote by $H_k^R$ the set:
\begin{align}
H_k^R = \left(G_0^{R \cap \{1\}}\right)^c \cap \left(G_1^{R \cap \{1,2\}}\right)^c \cap \dots \cap \left( G_{(k-1)}^{R \cap \{1,\dots,k\}}\right)^c.
\end{align}
In the cell $C_k^R$ with $\vert R \vert = j$, the trajectories experience $j-1$ inelastic collisions on $]0,k\delta]$, the last collision taking place in $]k\delta,(k+1)\delta]$. We find:
\begin{align}
&\left\vert G_k \cap L_{p,q}^Q \right\vert \leq \sum_{j=1}^{k+1} \Bigg\vert \bigcup_{\substack{ R \subset \{1,\dots,k+1\} \\ k+1 \in R \\ \vert R \vert = j}} \hspace{-6mm} \left( \left[ \mathcal{T}_{k\delta}^{-1}\left(P_k(\mu_j)\right) \cap H_k^R \cap C_k^R \right] \cap L_{p,q}^Q \right) \Bigg\vert.
\end{align}
Let us now use on the one hand that $\mathcal{T}_{k\delta}$ is injective on $L_{p,q}^Q$, and on the other hand that the cells $\left(C_k^R\right)_R$ form a partition of the initial configurations of $F_0^c \cap \dots \cap F_{(k-1)}^c$, on which in particular $\mathcal{T}_{k\delta}$ is defined. 
We deduce:
\begin{align}
\Bigg\vert \bigcup_{\substack{ R \subset \{1,\dots,k+1\} \\ (k+1) \in R \\ \vert R \vert = j}}& \hspace{-6mm} \left( \mathcal{T}_{k\delta}^{-1}\left(P_k(\mu_{k+1})\right) \cap C_k^R \cap H_k^R \cap L_{p,q}^Q \right) \Bigg\vert \nonumber\\
&\leq \sum_{\substack{ R \subset \{1,\dots,k+1\} \\ (k+1) \in R \\ \vert R \vert = j}} \Big\vert \mathcal{T}_{k\delta}^{-1}\left(P_k(\mu_j)\right) \cap C_k^R \cap H_k^R \cap L_{p,q}^Q \Big\vert \nonumber\\
&\leq \sum_{\substack{ R \subset \{1,\dots,k+1\} \\ (k+1) \in R \\ \vert R \vert = j}} \Big\vert \mathcal{T}_{k\delta}^{-1}\Big(P_k(\mu_j) \cap \mathcal{T}_{k\delta}\left(C_k^R \cap L_{p,q}^Q\right) \Big) \cap H_k^R \Big\vert \nonumber\\
&\hspace{50mm}\text{(because }\mathcal{T}_{k\delta}\text{ is injective on the cell }L_{p,q}^Q\text{)}\nonumber\\
&\leq \frac{1}{\sqrt{\mu_1\dots\mu_{j-1}}} \sum_{\substack{ R \subset \{1,\dots,k+1\} \\ (k+1) \in R \\ \vert R \vert = j}} \Big\vert P_k(\mu_j) \cap \mathcal{T}_{k\delta}\left(C_k^R \cap L_{p,q}^Q\right) \Big)\Big\vert \nonumber\\
&\hspace{-22mm}\text{(by Lemma \ref{LEMMEEstimJacob}, for initial configurations in }H_k^R\text{, the Jacobian determinant of }\mathcal{T}_{k\delta}\text{ is larger than }\sqrt{\mu_1\dots\mu_{j-1}})\nonumber\\
&\leq \frac{1}{\sqrt{\mu_1\dots\mu_{j-1}}} \Big\vert P_k(\mu_j) \cap \hspace{-6mm} \bigcup_{\substack{ R \subset \{1,\dots,k+1\} \\ (k+1) \in R \\ \vert R \vert = j}} \hspace{-6mm} \mathcal{T}_{k\delta} \left(C_k^R \cap L_{p,q}^Q\right) \Big\vert \leq \frac{1}{\sqrt{\mu_1\dots\mu_{j-1}}} \big\vert P_k(\mu_j) \big\vert.
\end{align}
In the last line, we used \eqref{EQUATTheorEnsmb2Decp__P__} in a crucial way to remove the sum over all the subsets $R$, relying once again on the fact that $\mathcal{T}_{k\delta}$ is injective on $L_{p,q}^Q$ and that the family $\left(C_k\cap L_{p,q}^Q\right)_R$ forms a partition of $(\cup_R C_k^R)\cap L_{p,q}^Q$.\\
In the end, we obtain:
\begin{align}
\left\vert \mathcal{P}(\delta,\mu_1,\dots,\mu_n) \right\vert &\leq \sum_{k=0}^{n-1} \left\vert G_k \right\vert \leq \sum_{p=0}^{\min(n,C_\text{IHS})}\sum_{q=0}^{\min(p,K)}\sum_{\substack{Q \subset \{1,\dots,p\}\\ \vert Q \vert = q}} \sum_{k=0}^{n-1} \left\vert G_k \cap L_{p,q}^Q \right\vert \nonumber\\
&\leq \sum_{p=0}^{\min(n,C_\text{IHS})}\sum_{q=0}^{\min(p,K)}\sum_{\substack{Q \subset \{1,\dots,p\}\\ \vert Q \vert = q}} \sum_{k=0}^{n-1} \sum_{j=1}^{k+1} \Bigg\vert \bigcup_{\substack{ R \subset \{1,\dots,k+1\} \\ k+1 \in R \\ \vert R \vert = j}} \hspace{-6mm} \left( \left[ \mathcal{T}_{k\delta}^{-1}\left(P_k(\mu_j)\right) \cap H_k^R \cap C_k^R \right] \cap L_{p,q}^Q \right) \Bigg\vert \nonumber\\
&\leq \sum_{p=0}^{\min(n,C_\text{IHS})}\sum_{q=0}^{\min(p,K)}\sum_{\substack{Q \subset \{1,\dots,p\}\\ \vert Q \vert = q}} \sum_{k=0}^{n-1} \sum_{j=1}^{k+1} \frac{1}{\sqrt{\mu_1\dots\mu_{j-1}}} \vert P_k(\mu_j) \vert \nonumber\\
&\leq \sum_{p=0}^{\min(n,C_\text{IHS})}\sum_{q=0}^{\min(p,K)} \sum_{\substack{Q \subset \{1,\dots,p\}\\ \vert Q \vert = q}} \sum_{k=0}^{n-1} \sum_{j=1}^{k+1} \frac{C(N,T,\varepsilon_0,R_1,R_2) \delta \mu_j}{\sqrt{\mu_1\dots\mu_{j-1}}} \cdotp
\end{align}
Finally, let us observe that the last estimate on the measure of $\mathcal{P}(\mu_1,\dots,\mu_K)$ can be improved in a dramatic way: although the sum over $k$ cannot be estimated better, the sum over $j$ can. Indeed, $j$ was introduced to consider all the possible cardinals of the subsets $R \subset \{1,\dots,k+1\}$. If we keep in mind that $R$ was used to label the time intervals $]l\delta,(l+1)\delta]$ in which the trajectories experienced inelastic collisions, we have actually $\vert R \vert \leq K$, so that:
\begin{align}
\label{EQUATCntrlMesur__P__}
\left\vert \mathcal{P}(\delta,\mu_1,\dots,\mu_K) \right\vert &\leq \sum_{p=0}^{\min(n,C_\text{IHS})}\sum_{q=0}^{\min(p,K)}\sum_{\substack{Q \subset \{1,\dots,p\}\\ \vert Q \vert = q}} \sum_{k=0}^{n-1} \sum_{j=1}^{\min(k+1,K)} \frac{C(N,T,\varepsilon_0,R_1,R_2) \delta \mu_j}{\sqrt{\mu_1\dots\mu_{j-1}}} \nonumber\\
&\leq C(N,T,\varepsilon_0,R_1,R_2) \sum_{j=1}^K \frac{\mu_j}{\sqrt{\mu_1\dots\mu_{j-1}}},
\end{align}
using in the last line that by definition $n = T/\delta$, so that $\sum_{k=0}^{n-1} \delta = T$, and that $C_\text{IHS}$ and $K$ depend only on $N$ and $R_2$ (as a consequence of the theorem of Burago-Ferleger-Kononenko \cite{BuFK998}).\\
Let us observe that we managed to construct a pathological set $\mathcal{P}(\mu_1,\dots,\mu_K)$ outside which the Jacobian determinant of the transport of the IHSE system is controlled in terms of the parameter $\mu_i$, and such that, according to \eqref{EQUATCntrlMesur__P__}, the measure of $\mathcal{P}(\mu_1,\dots,\mu_K)$ depends only on these parameters $\mu_i$. In particular, the measure of $\mathcal{P}(\mu_1,\dots,\mu_K)$ is independent from the time truncation parameter $\delta$. This is a crucial result to adapt Alexander's approach to the model we are considering in the present article.\\
\newline
\noindent
We can now estimate $\left\vert \mathcal{A}(\delta) \right\vert$. We have:
\begin{align}
\left\vert \mathcal{A}(\delta) \right\vert \leq \sum_{k=0}^{n-1} \left\vert F_k \right\vert.
\end{align}
To estimate the measure of each of the sets $F_k$, we perform a decomposition similar as for the study of $G_k$. More precisely, let us assume that $F_0,\dots,F_{(k-1)}$ and $G_0,\dots,G_{(k-1)}$ are defined, such that the transport $\mathcal{T}_{k\delta}$ is well-defined on $F_0^c \cap \dots \cap F_{(k-1)}^c$, and its Jacobian determinant is estimated outside $G_0\cup\dots\cup G_{(k-1)}$.\\
First, we decompose:
\begin{align}
F_k = \left[ F_k \cap \left(G_0\cup\dots\cup G_{(k-1)}\right)^c\right] \cup \left[ F_k \cap \left(G_0\cup\dots\cup G_{(k-1)}\right) \right].
\end{align}
The measure of the second term is of course estimated by the measure of $G_0\cup\dots\cup G_{(k-1)}$, so that we have:
\begin{align}
\bigcup_{k=0}^{n-1} \left[ F_k \cap \left(G_0\cup\dots\cup G_{(k-1)}\right) \right] \subset \bigcup_{k=0}^{n-1} \left[ F_k \cap \mathcal{P}(\mu_1,\dots,\mu_K) \right] \subset \mathcal{P}(\mu_1,\dots,\mu_K).
\end{align}
We will then focus on the first term. As for $G_k$ we decompose, first, on the cells $L_{p,q}^Q$, then on the cells $C_{(k-1)}^R$. We write:
\begin{align}
&\hspace{-20mm}\left\vert F_k \cap \left(G_0^c\cap\dots\cap G_{(k-1)}^c\right) \cap L_{p,q}^Q \right\vert \nonumber \\
&\leq \sum_{j=0}^k \sum_{\substack{R \subset \{1,\dots,k\} \\ \vert R \vert = j}} \Big\vert \mathcal{T}_{k\delta}^{-1} \left( E_k \cup \Delta \cap \mathcal{T}_{k\delta} \left(C_{(k-1)}^R \cap L_{p,q}^Q \right) \right) \cap \left(G_0^c\cap\dots\cap G_{(k-1)}^c\right) \Big\vert \nonumber\\
&\leq \sum_{j=0}^k \frac{1}{\sqrt{\mu_1\dots\mu_j}}\sum_{\substack{R \subset \{1,\dots,k\} \\ \vert R \vert = j}} \big\vert E_k \cup \Delta \cap \mathcal{T}_{k\delta}\left(C_{(k-1)}^R \cap L_{p,q}^Q\right) \big\vert \nonumber\\
&\leq \sum_{j=0}^k \frac{1}{\sqrt{\mu_1\dots\mu_j}} \big\vert E_k \cup \Delta \big\vert \nonumber\\
&\leq \sum_{j=0}^k \frac{C(N,T,R_1,R_2)\delta^2}{\sqrt{\mu_1\dots\mu_j}} \cdotp
\end{align}
Using that the total number of collisions is uniformly bounded, as well as the number of inelastic collisions, we find in the end:
\begin{align}
\left\vert \mathcal{A}(\delta) \right\vert \leq \sum_{p=0}^{\min(n,C_\text{IHS})}\sum_{q=0}^{\min(p,K)}\sum_{\substack{Q\subset\{1,\dots,p\} \\ \vert Q \vert = q}}\sum_{k=0}^{n-1} \sum_{j=0}^K \frac{C(N,T,R_1,R_2)\delta^2}{\sqrt{\mu_1\dots\mu_j}} + \left\vert \mathcal{P}(\mu_1,\dots,\mu_K) \right\vert,
\end{align}
so that, using finally that $\sum_{k=0}^{n-1}\delta^2 = T\delta$, we obtain:
\begin{align}
\left\vert \mathcal{A}(\delta) \right\vert \leq C(N,T,R_1,R_2)\sum_{j=0}^K \frac{\delta}{\sqrt{\mu_1\dots\mu_j}} + C(N,T,\varepsilon_0,R_1,R_2) \sum_{j=1}^K \frac{\mu_j}{\sqrt{\mu_1\dots\mu_{j-1}}} \cdotp
\end{align}
In order to conclude, it remains to choose the parameters $\mu_i$ in an appropriate manner. Let us observe that for $h>0$ small enough such that $h^{3/2} \leq 1/2$, if we define:
\begin{align}
\forall 1\leq i \leq K,\ \mu_i = h^{\left(3/2\right)^i},
\end{align}
we have, on the one hand:
\begin{align}
\mu_1 = \frac{\mu_2}{\sqrt{\mu_1}} = \dots = \frac{\mu_K}{\sqrt{\mu_1\dots\mu_{K-1}}} = h^{3/2},
\end{align}
and, on the other hand:
\begin{align}
\frac{1}{\sqrt{\mu_1\dots\mu_K}} = \frac{1}{h^{\frac{3}{2}\left(\left(\frac{3}{2}\right)^K-1\right)}}\cdotp
\end{align}
Therefore, if we choose $h$ as:
\begin{align}
h = \delta^{\left(\frac{3}{2}\right)^{K+1}},
\end{align}
then, up to multiplicative constants that depend only on $N$, $T$, $\varepsilon_0$, $R_1$ and $R_2$, we find the following equivalents for $\left\vert \mathcal{A}(\delta) \right\vert$ and $\left\vert \mathcal{P}(\delta,\mu_1,\dots,\mu_K) \right\vert$ as $\delta \rightarrow 0$, and :
\begin{align}
\left\vert \mathcal{A}(\delta) \right\vert, \left\vert \mathcal{P}(\mu_1,\dots,\mu_K) \right\vert \sim \delta^{\left(\frac{2}{3}\right)^K}.
\end{align}
We can now consider the intersection of the pathological sets $\vert \mathcal{A}(\delta) \vert$ and $\mathcal{P}(\mu_1,\dots,\mu_K)$ when $\delta$ is sent to zero, which provides a set of zero measure in $B_{X_N}(0,R_1) \times B_{V_N}(0,R_2)$, such that the flow of IHSE is defined on its complement on $[0,T]$. Repeating the argument for three countable sequences $\left(R_{1,n}\right)_n$, $\left(R_{2,n}\right)_n$ and $\left(T_n\right)_n$ that all tend to infinity as $n$ goes to infinity, we complete the proof of Theorem \ref{THEORSect3AlexanderIHSWEM}.
\end{proof}

\begin{remar}
Our proof of Theorem \ref{THEORSect3AlexanderIHSWEM} shows that an Alexander's-like result holds for any system of particles that can undergo only a finite number of inelastic collisions, uniformly on any compact set of the phase space, and such that the scattering $S$ is not shrinking too much the measure, namely, such that:
\begin{align}
\vert S(A) \vert \geq C_K \vert A \vert,
\end{align}
where $C_K$ is a local constant, that may depend on the compact set $K$ of the phase space on which we consider the restriction of the scattering mapping $S$. Indeed, to prove Alexander's theorem for the system of inelastic hard spheres with emission, we did not need that the measure was preserved.
\end{remar}

\begin{remar}
In the case of the IHSE system, when the kinetic energy in bounded, we proved that only a finite number of collisions can take place in the system, uniformly in time, and uniformly in the initial configurations. We used this result in a crucial way, to decompose the phase space in a finite number of cells, such that the flow was injective when restricted to such cells.\\
The absence of such a uniform upper bound is a fundamental difference between our model, and the classical inelastic model \eqref{EQUAT_Loi_ColliCoeffRestiRFixe}, with fixed restitution coefficient. In particular, since for such a model it can be proved that the inelastic collapse can take place, even for a very small kinetic energy, there is no hope to obtain an upper bound on the number of collisions in compact subsets of the phase space. Therefore the method used to prove Theorem \ref{THEORSect3AlexanderIHSWEM} does not apply to the inelastic hard sphere system with fixed restitution coefficient, and its well-posedness remains a challenging open question.
\end{remar}

\subsection{Direct proof of Theorem \ref{THEORSect3AlexanderIHSWEM} in the case of a single inelastic collision}

Let us note that the proof of Alexander's theorem can be directly adapted for systems of inelastic hard spheres with emission, provided that the total energy of the system is small. In particular, a proof can be provided without  without using the strong result of \cite{BuFK998} concerning the uniform number of collisions for a system of $N$ hard spheres.\\
However, the assumption on the smallness of the kinetic energy is quite stringent. Namely, if we assume that the initial energy is so small that only a single inelastic collision can take place, the proof that can be found in \cite{GSRT013} can be adapted in the following way. The main idea is to consider non-uniform time steps, that depend not only on the number $n$ of sub-intervals used to decompose $[0,T]$, but also on the $k$-th iteration of the TCT dynamics that allowed to define the transport. We obtain the following result, which is a particular case of Theorem \ref{THEORSect3AlexanderIHSWEM}.

\begin{theor}[Alexander's theorem for the inelastic hard sphere with emission model, in the low energy case]
\label{THEORSect3AlexanderIHSWEMFaiblEnerg}
Let $N$ be any positive integer, and let $\varepsilon_0 > 0$ be a positive real number. On the subset of the phase space $\mathcal{D}_N$ such that
\begin{align}
\label{EQUATHypotFaiblEnerg}
\sum_{k=1}^N \frac{\vert v_k \vert^2}{2} < 2 \varepsilon_0,
\end{align}
the same results as in Theorem \ref{THEORSect3AlexanderIHSWEM} hold.

\end{theor}

\begin{remar}
The low energy condition \eqref{EQUATHypotFaiblEnerg} implies that at most one inelastic collision can take place.\\
Indeed, after the first inelastic collision, the total kinetic energy after such a collision becomes smaller than $\varepsilon_0$. We denote by $w_k$ ($1 \leq k \leq N$) the velocities of the particles of the system after this first inelastic collision. In this case, for any pair of particles, the relative velocity $w_i-w_j$ after this first inelastic collision satisfies:
\begin{align}
\vert w_i - w_j \vert^2 \leq 2 \left( \vert w_i \vert^2 + \vert w_j \vert^2 \right) \leq 4 \left( \sum_{k=1}^N \frac{\vert w_k \vert^2}{2} \right) < 4 \varepsilon_0.
\end{align}
Therefore, no inelastic collision can further take place in the particle system.
\end{remar}

\begin{proof}[Proof of Theorem \ref{THEORSect3AlexanderIHSWEMFaiblEnerg}]

We consider:
\begin{align}
\label{EQUATChoixDelta}
\delta_n(k) = \frac{T}{k \ln(n)},
\end{align}
together with, for $k \geq 2$:
\begin{align}
\label{EQUATChoix_Mu__}
\mu_n(k) = \alpha\frac{\ln^p(k)}{\ln^q(n)}, \hspace{5mm} \text{with} \hspace{5mm} 0 < p < q < 1,
\end{align}
and $\alpha$ small enough such that $\mu_n(k) \leq 1/2$.
In that case, considering $k$ iterations of the process allows to reach the final time, as $n\rightarrow +\infty$:
\begin{align}
\sum_{k=1}^n \delta_n(k) = \sum_{k=1}^n \frac{T}{k \ln(n)} \sim T.
\end{align}
As in the proof of Theorem \ref{THEORSect3AlexanderIHSWEM}, we eliminate step by step subsets $F_{0+\delta_1+\dots+\delta_k}$ and $G_{0+\delta_1+\dots+\delta_k}$ of initial configurations that allow in their complement, respectively, to consider a flow of particle which is well-defined until time $\delta_1+\dots+\delta_{k+1}$ on the one hand, and such that the Jacobian determinant of such a flow is not too small on the other hand. However, in the case when at most one inelastic collision can take place for any trajectory, the decomposition presented in proof of Theorem \ref{THEORSect3AlexanderIHSWEM} simplifies.\\
We write $\sum_{j=0}^k \delta_j = m_k$. The definition of the sets $F_{0+\delta_1+\dots+\delta_k} = F_{m_k}$ (taking as a convention that $\delta_0 = 0$), outside which the flow will be defined, is similar. We define the first subset as:
\begin{align}
F_0 = E_0 \cup \Delta,
\end{align}
outside which the flow is defined on the time interval $]0,\delta_1]$ (where $E_0$ is defined as in \eqref{EQUATDefinEnsmbPatho_E_kd}), and the others as:
\begin{align}
F_{m_k} = \mathcal{T}_{m_k}^{-1}\left( E_{m_k} \right) \cap F_0^c \cap \dots \cap F_{m_{k-1}}^c,
\end{align}
outside which the flow is defined on the time interval $]0,m_{k+1}] =\ ]0,\sum_{j=0}^{k+1} \delta_j]$. One of the main simplifications takes place for the sets $G_{m_k}$, we define them recursively as well. The first set will be:
\begin{align}
G_0 = P_0(\mu_n(1)),
\end{align}
set outside which the Jacobian determinant of $\mathcal{T}_{\delta_1}$ is larger than $\sqrt{\mu_1}$ (where $P_0(\mu_1)$ is defined as in \eqref{EQUATDefinEnsmbPatho_P_kd}). For $k \geq 1$, as a main difference with the general case treated in Theorem \ref{THEORSect3AlexanderIHSWEM}, it has to be observed that to define the sets $G_{m_k}$, we do not need to decompose the set $ F_0^c \cap \dots \cap F_{m_{k-1}}^c$. The reason is that we do not need to separate the configurations depending on the number of inelastic collisions they undergo on the time interval $]0,m_k] =\ ]0,\sum_{j=0}^k \delta_j]$. Indeed, the set $P_{m_k}$ is introduced to estimate the Jacobian of the TCT dynamics associated to the trajectory for the time interval $]m_k,m_{k+1}] =\ ]\sum_{j=0}^k \delta_j,\sum_{j=0}^{k+1} \delta_j]$, when such a TCT dynamics describes an inelastic collision. But since by assumption a single inelastic can take place on the whole time interval $]0,T]$ for any trajectory, it is relevant to consider only the trajectories reaching $P_{m_k}$ without any inelastic collision on the time interval $]0,m_k]$. This means that there is ``a single preimage'' of $P_{m_k}$ to be considered among the pathological initial configurations, the one obtained with the classical elastic hard sphere transport (that is, the one obtained only with free flow or elastic TCT dynamics). As a consequence, and reusing the notations introduced in the proof of Theorem \ref{THEORSect3AlexanderIHSWEM} concerning the cells $C_k^R$, we define, for $k\geq 1$:
\begin{align}
G_{m_k} = \mathcal{T}_{m_k}^{-1}\left( P_{m_k}(\mu_{k+1}) \right) \cap F_0^c \cap \dots \cap F_{m_{k-1}}^c \cap C_{m_{k-1}}^\emptyset.
\end{align}
Defining now the set $\mathcal{P}(n)$ as:
\begin{align}
\mathcal{P}(n) = \bigcup_{k=0}^{n-1} G_{m_k},
\end{align}
we can directly estimate its measure, because $\mathcal{T}_{m_k}$ is the elastic hard sphere transport on $C_{m_{k-1}}^\emptyset$, so it preserves the measure, and we have:
\begin{align}
\left\vert \mathcal{P}(n) \right\vert &\leq \sum_{k=0}^{n-1} \left\vert \mathcal{T}_{m_k}^{-1}\left( P_{m_k}(\mu_{k+1}) \right) \cap C_{m_{k-1}}^\emptyset \right\vert \nonumber\\
&\leq \sum_{k=0}^{n-1} \left\vert P_{m_k}(\mu_{k+1}) \right\vert \nonumber\\
&\leq C \sum_{k=1}^{n} \delta_n(k) \mu_n(k), 
\end{align}
where $C$ is a constant that depends only on $N$, $T$, $\varepsilon_0$, $R_1$ and $R_2$.\\
As for $\mathcal{A}(n)$, the union of the sets $F_{m_k}$, as for the proof of Theorem \ref{THEORSect3AlexanderIHSWEM}, we have to consider a decomposition of the initial trajectories, depending when the inelastic collisions take place. Nevertheless, using again the fact that a single inelastic collision can take place for any trajectory, the number of possibles cells is much smaller than in the general case. We perform the following decomposition. For the second set $F_{\delta_1}$, we write:
\begin{align}
F_{\delta_1} &= \mathcal{T}_{\delta_1}^{-1}(E_{\delta_1}) \cap F_0^c \nonumber\\
&= \left( \mathcal{T}_{\delta_1}^{-1}(E_{\delta_1}) \cap F_0^c \cap P_0(\mu_n(1)) \right) \nonumber\\
&\hspace{20mm}\cup \Big[ \left( \mathcal{T}_{\delta_1}^{-1}(E_{\delta_1}) \cap F_0^c \cap P_0(\mu_n(1))^c \cap C_0^\emptyset \right) \cup \left( \mathcal{T}_{\delta_1}^{-1}(E_{\delta_1}) \cap F_0^c \cap P_0(\mu_n(1))^c \cap C_0^{\{1\}} \right) \Big],
\end{align}
decomposing in three terms, the first that will be absorbed in $\mathcal{P}(n)$, while the second and third terms correspond respectively to the preimages (outside $P_0$) of the pathological set $E_{\delta_1}$, obtained by the elastic hard sphere flow and by an inelastic TCT dynamics respectively. More generally we decompose:
\begin{align}
F_{m_k} &= \Big( F_{m_k} \cap \bigcup_{l=0}^{k-1} G_{m_l} \Big) \nonumber\\
&\hspace{2mm} \cup \Big[ \Big( F_{m_k} \cap \bigcap_{l=0}^{k-1} G_{m_l}^c \cap C_{m_{k-1}}^\emptyset \Big) \cup \bigcup_{i=1}^k \Big( F_{m_k} \cap \bigcap_{l=0}^{k-1} G_{m_l}^c \cap C_{m_{k-1}}^{\{i\}} \Big) \Big],
\end{align}
so that:
\begin{align}
\left\vert \mathcal{A}(n) \right\vert &\leq \bigg\vert \bigcup_{k=0}^{n-1} \Big( F_{m_k} \cap \bigcup_{l=0}^{k-1} G_{m_l} \Big) \bigg\vert \nonumber\\
&\hspace{5mm} + \bigg\vert \bigcup_{k=0}^{n-1} \Big( F_{m_k} \cap C_{m_{k-1}}^\emptyset \Big) \bigg\vert  +  \bigg\vert \bigcup_{k=0}^{n-1}\bigcup_{i=1}^k \Big( F_{m_k} \cap \bigcap_{l=0}^{k-1} G_{m_l}^c \cap C_{m_{k-1}}^{\{i\}} \Big) \bigg\vert \nonumber\\
&\leq \left\vert \mathcal{P}(n) \right\vert + \sum_{k=0}^{n-1} \Big\vert E_{m_k} \Big\vert + \sum_{k=0}^{n-1} \sum_{j=1}^k \frac{1}{\sqrt{\mu_n(j)}} \Big\vert E_{m_k} \Big\vert \nonumber\\
&\leq \left\vert \mathcal{P}(n) \right\vert + \sum_{k=1}^n \sum_{j=0}^k \frac{1}{\sqrt{\mu_n(j)}} C \delta_n(k)^2,
\end{align}
where $C$ denotes again a constant that depends only on $N$, $T$, $\varepsilon_0$, $R_1$ and $R_2$, and where $\mu_n(0) = 1$ by convention.\\
There are two series to estimate. By replacing $\delta_n(k)$ and $\mu_n(j)$ with the particular values we have chosen in \eqref{EQUATChoixDelta} and \eqref{EQUATChoix_Mu__}, we find on the one hand:
\begin{align}
\label{EQUATPremiSommeAlexaSsBFK}
\sum_{k=2}^n \delta_n(k) \mu_n(k) = \frac{\alpha T}{\ln^{q+1}(n)}\sum_{k=2}^n \frac{\ln^p(k)}{k} \underset{n \rightarrow +\infty}{\sim} \frac{\alpha T}{\ln^{q+1}(n)} \frac{\ln^{p+1}(n)}{p+1},
\end{align}
and on the other hand:
\begin{align}
\sum_{k=1}^n\sum_{j=2}^k \frac{1}{\sqrt{\mu_n(j)}} \delta_n(k)^2 = \frac{T^2}{\sqrt{\alpha}} \frac{\ln^{q/2}(n)}{\ln^2(n)} \sum_{k=1}^n\sum_{j=2}^k \frac{1}{\ln^{p/2}(j)}\frac{1}{k^2},
\end{align}
where the last sum can be estimated by decomposing the summation on $k$, between $1$ and $\ln(n)$ on the one hand, and between $\ln(n)$ and $n$ on the other hand. Using the equivalent:
\begin{align}
\sum_{k=2}^{n} \frac{1}{\ln^p(k)} \underset{n\rightarrow+\infty}{\sim}\frac{n}{\ln^p(n)}
\end{align}
provides in the two following sums:
\begin{align}
\sum_{k=1}^{\ln(n)}\sum_{j=2}^k \frac{1}{\ln^{p/2}(j)}\frac{1}{k^2} \leq \frac{\pi^2}{6} \sum_{j=2}^{\ln(n)} \frac{1}{\ln^{p/2}(j)} \underset{n\rightarrow+\infty}{\sim} \frac{\pi^2}{6}\frac{\ln(n)}{\ln^{p/2}(\ln(n))},
\end{align}
and
\begin{align}
\sum_{k=\ln(n)}^n \sum_{j=2}^k \frac{1}{\ln^{p/2}(j)} \frac{1}{k^2} \underset{n\rightarrow+\infty}{\sim} \sum_{k=\ln(n)}^n \frac{1}{k \ln^{p/2}(k)} \underset{n\rightarrow+\infty}{\sim} \frac{1}{1-p/2} \ln^{1-p/2}(n),
\end{align}
so that, for $n$ large:
\begin{align}
\label{EQUATSecndSommeAlexaSsBFK}
\sum_{k=1}^n\sum_{j=2}^k \frac{1}{\sqrt{\mu_n(j)}} \delta_n(k)^2 \leq \frac{T^2}{\sqrt{\alpha}} \left( \frac{\pi^2}{6} \frac{1}{\ln^{1-q/2}\ln^{p/2}(\ln(n))} + \frac{1}{1-p/2} \frac{1}{\ln^{1+(p-q)/2}(n)} \right).
\end{align}
\eqref{EQUATPremiSommeAlexaSsBFK} and \eqref{EQUATSecndSommeAlexaSsBFK} enable to deduce that $\left\vert \mathcal{A}(n) \right\vert$ and $\left\vert \mathcal{P}(n) \right\vert$ are indeed vanishing quantities as the number of the time steps $n$ goes to infinity, so that we can conclude the proof of Alexander's theorem, in the case when any trajectory undergoes at most one inelastic collision.
\end{proof}

\subsection{Interpretation of the result of Theorem \ref{THEORSect2ConservatiMesur}}

The result of Theorem \ref{THEORSect2ConservatiMesur} might look surprising at first glance, because the scattering mapping we introduced in \eqref{EQUATSect1LoideCollision_} does preserve the measure in the phase space, although a positive amount of kinetic energy is lost at each inelastic collision. However, there is a way to interpret the result, in terms of conformal mappings. Denoting indeed by $m$ and $w$ the respective quantities:
\begin{align}
m = \frac{1}{2}(v+v_*),\hspace{5mm} w = v_*-v,
\end{align}
the collision mapping \eqref{EQUATSect1LoideCollision_} can be rewritten, in frame of the center of mass, as:
\begin{align}
(m,w) \mapsto \left(m,2\sigma(w) \sqrt{ \left\vert \frac{w}{2} \right\vert^2 - \varepsilon_0}\right),
\end{align}
where $\sigma(w)$ is the symmetry of the unit vector $w/\vert w \vert$ with respect to the orthogonal of $\omega = (x_*-x)/\vert x_*-x \vert$. Now, if we use the polar coordinates, choosing any unit orthogonal vector to $\omega$ and defining such a vector as the first axis (i.e., of angle $0$), the mapping:
\begin{align*}
w \mapsto 2\sigma(w) \sqrt{ \left\vert \frac{w}{2} \right\vert^2 - \varepsilon_0 }
\end{align*}
can be rewritten as:
\begin{align}
f:(\rho,\theta) \mapsto \left(f_\rho(\rho,\theta),f_\theta(\rho,\theta)\right) = \left( \sqrt{ \rho^2 - 4\varepsilon_0 },-\theta \right).
\end{align}
The infinitesimal element of area $[\rho_0,\rho_0+\dd \rho] \times [\theta_0,\theta_0+\dd\theta]$, written in polar coordinates, has the infinitesimal surface $\left\vert \rho_0 \dd \rho \dd \theta \right\vert$, and its image has the infinitesimal surface:
\begin{align*}
\left\vert f_\rho(\rho_0,\theta_0) \cdot \partial_\rho f_\rho(\rho_0,\theta_0) \dd\rho \cdot \partial_\theta f_\theta(\rho_0,\theta_0)\dd \theta \right\vert.
\end{align*}
But since we have:
\begin{align}
f_\rho(\rho_0,\theta_0) \partial_\rho f_\rho(\rho_0,\theta_0) \partial_\theta f_\theta(\rho_0,\theta_0) = \sqrt{ \rho_0^2-4\varepsilon_0 } \frac{\rho_0}{\sqrt{ \rho_0^2-4\varepsilon_0 }} (-1) = - \rho_0,
\end{align}
we deduce that the mapping $f$ preserves the measure, and so does the scattering mapping \eqref{EQUATSect1LoideCollision_}.\\
\newline
Such a computation can also be carried on in larger dimension using spherical coordinates, which provides similar collision models, with scatterings that preserve the measure in the phase space, but that describe a loss of some kinetic energy during the collisions. In dimension $d=3$ for instance, considering the mapping $f$, written in spherical coordinates $\left(\rho,\theta,\varphi\right)$ (with $\theta \in [-\pi/2,\pi/2]$ and $\varphi \in [0,2\pi]$) such that:
\begin{align}
\label{EQUATSect3Exemple_Dim_d>2}
f:(\rho,\theta,\varphi) \mapsto \left(\left(\rho^3 - 4\varepsilon_0\right)^{1/3},-\theta,\varphi\right),
\end{align}
the image by $f$ of the infinitesimal volume $[\rho_0,\rho_0+\dd\rho]\times[\theta_0,\theta_0+\dd\theta]\times[\varphi_0,\varphi_0+\dd\varphi]$ has the measure:
\begin{align}
\left((f_\rho)^2 \partial_\rho f_\rho \dd \rho\right) \cdot (\cos\theta_0 \partial_\theta f_\theta \dd \theta) \cdot (\partial_\varphi f_\varphi \dd \varphi) &= \left( \rho_0^3 - 4\varepsilon_0 \right)^{2/3} \frac{3\rho_0^2}{3\left( \rho_0^3 - 4\varepsilon_0 \right)^{2/3}} \cdot (\cos\theta_0 \partial_\theta f_\theta \dd \theta) \cdot (\partial_\varphi f_\varphi \dd\varphi) \nonumber \\
&= \rho_0^2 \cos\theta_0 \dd \theta \dd \varphi.
\end{align}
The mapping $f$ preserves the measure, and the same conclusion follows accordingly for the scattering mapping \eqref{EQUATSect1LoideCollision_} in dimension $3$. It is clear that one can construct such examples for any arbitrary dimension. Let us observe that although such scattering mappings conserve the measure in the phase space, and describe a loss of a positive quantity of kinetic energy in any collision that is energetic enough, the way in which the particles lose kinetic energy in the particle model we consider here (defined using \eqref{EQUATSect3Exemple_Dim_d>2}) becomes less natural than \eqref{EQUATSect1DissipatioEnCin}, that holds only in dimension $d=2$.\\
\newline

\noindent
\textbf{Acknowledgements.} The authors are grateful to I. Gallagher for very interesting discussions, comments and questions that led us to write this paper, and for having discovered a flaw in the original version of the argument, which enabled us to elaborate the final version of the present article.\\
The authors are also grateful to the anonymous referee, who asked us to clarify an important point which led us to understand in a deeper way the TCT dynamics.\\
The authors gratefully acknowledge the financial support of the Hausdorff Research Institute for Mathematics (Bonn) through the collaborative research center The mathematics of emerging effects (CRC 1060, Project-ID 211504053), and the Deutsche Forschungsgemeinschaft (DFG, German Research Foundation). The first author also gratefully acknowledges the financial support of the project PRIN 2022 (Research Projects of National Relevance) - Project code 202277WX43, based at the University of L’Aquila.\\
The authors declare that the data supporting the findings of this study are available within the paper.

E-mail address: \texttt{theophile.dolmaire@univaq.it}, \texttt{velazquez@iam.uni-bonn.de}.

\end{document}